\documentclass[english]{article} 
\pdfoutput=1

\usepackage{amsmath, amssymb, graphicx, url, algorithm2e}

\usepackage{times}
\usepackage{graphicx}
\usepackage{psfrag}
\usepackage{caption}
\usepackage{subcaption}
\usepackage{url}

 \usepackage{blkarray}
 \usepackage{multirow}
 \usepackage[table]{xcolor}

\usepackage{color}
\usepackage{amsthm}

\newtheorem{question}{Question}
\newtheorem{problem}{Problem}

\newtheorem{theorem}{Theorem}

\newtheorem{example}{Example}
\newtheorem{lemma}[theorem]{Lemma}
\newtheorem{corollary}[theorem]{Corollary}
\newtheorem{observation}{Observation}

\theoremstyle{plain}
\newtheorem{definition}[theorem]{Definition}

 \newcommand{\suppx}[2][D]{\ensuremath{\mathfrak{S}_{#1}({#2})}}
 \newcommand{\sxd}{\ensuremath{S_{X,D}}}
 \newcommand{\hsxd}{\ensuremath{H(S_{X,D})}}
 \newcommand{\cald}{\ensuremath{\mathfrak{I}(D)}}
 \newcommand{\calsxd}{\ensuremath{\mathfrak{I}(\sxd)}}

\newcommand{\R}[1][d]{\ensuremath{\mathbb{R}^{#1}}}
\newcommand{\proj}[1][d-1]{\ensuremath{\mathbb{P}^{#1}(\mathbb{R})}}
\newcommand{\algesystem}{\ensuremath{(H,Q)(p)}}
\newcommand*{\Cdot}{{\scalebox{1.25}{$\,\cdot\,$}}}
\newcommand{\framework}[1][p]{\ensuremath{(H,Q,#1)}}
\newcommand{\scalemath}[2]{\scalebox{#1}{\mbox{\ensuremath{\displaystyle #2}}}}

\usepackage{multirow}
\usepackage{tabu}

\begin{document}
\title{Combinatorial rigidity of Incidence systems and Application to Dictionary learning\thanks{This research was supported in part by the research grant NSF CCF-1117695 and a research gift from SolidWorks. Portions of this paper have appeared on Proceedings of 2014 Canadian Conference on Computational Geometry and Post-Proceedings of 2014 International Workshop on Automated Deduction in Geometry.}}

 
  \author{Meera Sitharam
  	\and Mohamad Tarifi
  	\and Menghan Wang} 
  

\maketitle

\begin{abstract}

Given a hypergraph $H$ with $m$ hyperedges and a set $Q$ of $m$ \emph{pinning subspaces},
i.e.\ globally fixed subspaces in Euclidean space $\R$,
a \emph{pinned subspace-incidence system} is the pair $(H, Q)$,
with the constraint that each pinning subspace in $Q$ is contained in the subspace spanned
by the point realizations in $\R$
of vertices of the corresponding hyperedge of $H$.
This paper
provides a combinatorial characterization of pinned subspace-incidence systems
that are \emph{minimally rigid},
i.e.\ those systems that are guaranteed to generically yield a locally
unique realization. 

Pinned subspace-incidence systems have applications in the \emph{Dictionary Learning (aka sparse coding)} problem,
i.e.\ the problem of obtaining a sparse representation of a given set of  data vectors by learning \emph{dictionary vectors} upon which  the data vectors can be written as sparse linear combinations.
Viewing the dictionary vectors from a geometry perspective as the spanning set of a subspace arrangement,
the result gives a tight bound on the number
of dictionary vectors for sufficiently randomly chosen data vectors, and
gives a way of constructing a dictionary that meets the bound.   For less
stringent restrictions on data, but a natural modification of the
dictionary learning problem, a further dictionary learning algorithm is
provided.
Although there are recent rigidity based approaches for low rank matrix completion,
we are unaware of prior application of combinatorial rigidity techniques 
in the setting of Dictionary Learning.  
We also provide a systematic classification of problems related to dictionary learning 
together with various algorithms, their assumptions  and performance. 
\end{abstract}

%

\pagestyle{myheadings}
\thispagestyle{plain}

\section{Introduction}

A \emph{pinned subspace-incidence system} $(H,Q)$ is an incidence constraint system
specified as a hypergraph $H$ together with a set $Q$ of \emph{pinning subspaces} in $\R$,
each specified by a collection of basis vectors called \emph{pins}.
The pinning subspaces are in one-to-one correspondence with the hyperedges of $H$.
A realization of $(H,Q)$ is a subspace arrangement that assigns vectors in $\R$ to the vertices of $H$.
Each hyperedge of $H$ is assigned the subspace spanned by its vertex vectors and contains the associated pinning subspace in $Q$.

Pinned subspace-incidence systems naturally arise in finding bounds for the \emph{dictionary learning (aka sparse coding)} problem~\cite{olshausen1997sparse} for highly general data,
i.e.\ the problem of obtaining a sparse representation of data vectors by learning \emph{dictionary vectors} upon which the data vectors can be written as sparse linear combinations.
This set of dictionary vectors can be viewed from a geometry perspective as the spanning set of a subspace arrangement,
where the subspaces contain the data vectors.
The subset of data vectors on a given subspace spans a pinning subspace.
Thus the solution of dictionary learning problem corresponds to a pinned subspace-incidence system.

The special case of 2-dimensional pinned line-incidence systems was also used in modeling microfibirils in biomaterials such as cellulose and collagen~\cite{baker2015designing}.
Each such microfibril is attached to some fixed larger organelle/membrane at one site 
and  cross-linked at two sites with other fibrils,
where the \emph{cross-linking} is like an incidence constraint that the
crosslinked fibrils can slide against each other while remaining incident.
Consequently,  they can be modeled using a 
pinned line-incidence system with $H$ being a graph, 
where each fibril is modeled as an edge of $H$ with the two cross-linkings as its two vertices,
and the attachment is modeled as the corresponding pin.

Previous works on related types of frameworks
include pin-collinear body-pin frameworks~\cite{jackson2008pin}, 
direction networks~\cite{whiteley1996some},
slider-pinning rigidity~\cite{streinu2010slider},
body-cad constraint system~\cite{haller2012body},
$k$-frames~\cite{white1987algebraic,white1983algebraic}, 
and affine rigidity	~\cite{gortler2013affine}.
All of which involve some form of incidence constraints.
However, we are not aware of any previous results on systems that are similar to pinned subspace-incidence systems.

\begin{figure}
	\centering
		\includegraphics[width=.4\linewidth]{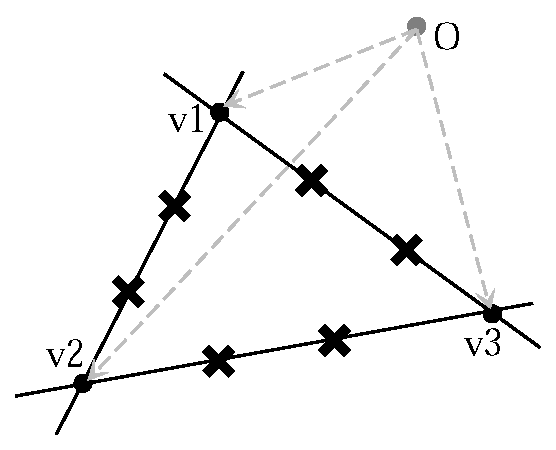}
		\caption{A pinned subspace-incidence framework with $d=3$, projectivized in $\proj[2]$.}
		\label{fig:pinned_hypergraph}
\end{figure}

\section{Contributions}

In this paper, we follow the combinatorial rigidity approaches \cite{asimow1978rigidity,white1987algebraic}
to give a complete combinatorial characterization for pinned subspace-incidence systems.
Specifically, 
\begin{itemize}
\item We formulate the pinned subspace-incidence systems 
	as a nonlinear algebraic system $\algesystem$ and apply  classic method of Asimow and Roth \cite{asimow1978rigidity} by linearizing  $\algesystem$.
\item We then apply another well-known method of White and Whiteley \cite{white1987algebraic} to combinatorially characterize the rigidity of the underlying hypergraph $H$,  using the Laplace decomposition of the \emph{rigidity matrix}, 
which corresponds to a \emph{map-decomposition}~\cite{streinu2009sparse} of the underlying hypergraph. 
The  polynomial resulting from the Laplace decomposition is called the
\emph{pure condition}, which characterizes the conditions that the framework has 
to avoid for the combinatorial characterization to hold.	
\end{itemize}

To our best knowledge, 
the only known results with a similar flavor are \cite{haller2012body,lee-st.john2013combinatorics} which characterize the rigidity of Body-and-Cad frameworks. 
However, these results are dedicated to specific frameworks in 3D instead of arbitrary dimension subspace arrangements and hypergraphs, 
and their formulation process start directly with the linearized Jacobian (omitting the first bullet alone). 

We then apply the combinatorial rigidity result to  dictionary learning problems.
\begin{itemize}
	\item As a corollary of the main result, we give a tight bound on the number
	of dictionary vectors for sufficiently randomly chosen data vectors, and
	give a way of constructing a dictionary that meets the bound
(see Corollary~\ref{cor:dictionary_size_bound} and \ref{corr:straight_forward_algo}).

	\item 	 
	%
	On the other hand, more common types of data can be handled using 
	a standard preprocessing step such as generalized PCA~\cite{vidal2005generalized}
	that converts the dictionary learning problem to a so-called \emph{fitted dictionary learning problem}
	that yields a specific pinned subspace incidence system where a realization yields a dictionary,
	followed by recursive decomposition of the underlying pinned subspace-incidence system (referred to as \emph{DR-planning}) to obtain the realization.
	
	%
	\item We also provide a systematic classification of problems related to Dictionary Learning 
	together with various approaches, assumptions required and performance
	(see Section \ref{sec:review}).   
	%
\end{itemize}
There are some recent applications of
rigidity in machine learning~\cite{jackson2014combinatorial,kiraly2012algebraic,singer2010uniqueness,liang2012identifiability}, 
specifically for low rank
matrix completability. These use the graph of entries of distance matrices
and gram matrices, i.e., rigidity with respect to distance and  inner
product constraints. We are however unaware of applications of rigidity
with respect to incidence constraints, or to dictionary learning.

\smallskip



\section{Preliminaries}
\label{sec:def}

In this section, we introduce the formal definition of pinned subspace-incidence systems
and basic concepts in combinatorial rigidity. 

A \emph{hypergraph} $H=(V,E)$ is a set $V$ of vertices and a set $E$ of hyperedges, 
where each hyperedge is a subset of $V$. 
The \emph{rank} $r(H)$ of a hypergraph $H$ is the maximum cardinality of any edge in $E$,
i.e.\ $r(H) = \max_{e_k \in E} s(e_k)$, where $s(e_k)$ denotes the cardinality of the hyperedge $e_k$.
A hypergraph is \emph{$s$-uniform} if all edges in $E$ have the same cardinality $s$, where $2$-uniform hypergraph is called a \emph{graph} $G$.
A \emph{configuration} or \emph{realization} of a hypergraph $H = (V,E)$ in  $\R$
is a mapping from the vertices of $H$ to the vectors in $\R$, i.e.\ $p: V \rightarrow \R$. 
When there is no ambiguity, we simply use $p_i$ to denote the vector $p(v_i)$, 
$p(e_k)$ to denote the set of vectors $\{p(v_i) | v_i \in e_k\}$,
and $s_k$ to denote the cardinality $s(e_k)$.

In the following,  we use $\langle P \rangle$ to denote the subspace \emph{spanned} by  a set $P$ of points in $\R$.

\begin{definition} [Pinned Subspace-Incidence System]
	A \emph{pinned subspace-incidence system} in  $\R$ is a pair $(H, Q)$, 
	where $H=(V,E,m)$ is a weighted hypergraph of rank  $r(H) < d$,
	and $Q = \{q_1, q_2, \ldots, q_{|E|}\}$ is a set of \emph{pinning subspaces} (subspaces of $\R$)
	in one-to-one correspondence with the hyperedges of $H$.
	Here the weight assignment is a function $m: E \rightarrow \mathbb{Z}^+$, 
	where  $m(e_k)$ denotes the dimension of the  pinning subspace $q_k$ associated with the hyperedge $e_k$. 
	We may write $m(q_k)$ or simply $m_k$ in substitute of $m(e_k)$.
	Often we ignore the weight $m$ and just refer to the hypergraph $(V,E)$ as $H$.
	
	A \emph{pinned subspace-incidence framework  realizing} the pinned subspace-incidence system  $(H,Q)$ is a triple $\framework$, 
	where $p$ is a realization of $H$,
	such that for all pinning subspaces $q_k \in Q$, 
	$q_k$  is contained in $\langle p(e_k) \rangle$, the subspace spanned by the set of vectors realizing the vertices of the  hyperedge $e_i$ corresponding to $q_k$.
	
\end{definition}

Since we only care about incidence relations, 
we projectivize the Euclidean space $\R$ to treat the pinned subspace-incidence system in the real projective space $\proj$,
and refer to the vectors $p_i$ realizing the vertices of $H$ as points in $\proj$ using the same notation
when the meaning is clear from the context.
Figure~\ref{fig:pinned_hypergraph} gives an example of a pinned subspace-incidence framework in the real projective space $\proj$ with $d=3$, where
the crosses denote the projectivized pins.
As each pinning subspace is spanned by two pins, $m(e) = 2$ for any hyperedge $e$.

\smallskip
\noindent\textbf{Note:} as the pinning subspaces in $Q$ are treated as globally fixed, the \emph{trivial motion group} of a pinned subspace-incidence system $(H,Q)$ reduces to the identity.

\begin{definition}
	A pinned subspace-incidence system $(H,Q)$ is \emph{independent} if none of the polynomial constraints is in the real ideal generated by others, 
	implying existence of a realization.
	It is \emph{rigid} if there exist at most finitely many realizations. 
	It is \emph{minimally rigid} if it is both rigid and independent. 
	It is \emph{globally rigid} if there exists at most one realization. 
	%
\end{definition}

\section{Algebraic Representation and Linearization}

In the following,  
we use $A[R,C]$ to denote a submatrix of a matrix $A$, 
where $R$ and $C$ are respectively index sets of the rows and columns contained in the submatrix.
In addition, $A[R,\Cdot]$ represents the submatrix containing row set $R$ and all columns, 
and $A[\Cdot,C]$ represents the submatrix containing column set $C$ and all rows.

\subsection{Representation Using Polynomials}

For any hyperedge  $e_k = \{v^k_1, v^k_2, \ldots, v^k_{|e_k|}\}$,
 the subspace $\langle p(e_k) \rangle$ spanned by the point set  $ \{p^k_1, p^k_2, \ldots, p^k_{|e_k|}\}$
is constrained to contain the pinning subspace $q_k$ associated with $e_k$.  
%
Recall that $q_k$ is a subspace of dimension $m_k - 1$ in $\proj$ spanned by a set of $m_k$ pins $\{x^k_1, x^k_2, \ldots, x^k_{m_k}\}$,
and the incidence constraint is equivalent to requiring $\langle p(e) \rangle$ to contain each pin $x^k_l$, for  $1 \le l \le m_k$.

\sloppy Using homogeneous coordinates
$p^k_i = [\begin{array}{cccc}p^k_{i,1} & p^k_{i,2} & \ldots & p^k_{i,d-1}\end{array}] $ and
$x^k_l = [ \begin{array}{cccc}x^k_{l,1} & x^k_{l,2} & \ldots & x^k_{l,d-1} \end{array}]$,
we write this incidence constraint for each pin $x^k_l$ 
by letting 
all the $|e_k| \times |e_k|$ minors of the $|e_k| \times (d-1) $ matrix
\[
E^k_l=
\left[
\begin{array}{c}
p^k_1 - x^k_l	\\
p^k_2 -x^k_l \\
\vdots \\
p^k_{|e_k|} - x^k_l 
\end{array} \right]
\]
\fussy
be zero. There are $d-1 \choose |e_k|$ minors, giving $d-1 \choose |e_k|$ equations. 
%
Note that any $d-|e_k|$ of these $d-1 \choose |e_k|$ equations are independent and span the rest.
So we can write the incidence constraint as $(d-|e_k|)$ independent equations: 
\begin{equation} \label{eq:constraint}
\det\left(E^k_l[\Cdot, C(t)]\right) = 0, \qquad 1 \le t \le d-|e_k|
\end{equation} 
where $C(t)$ denote the following index sets of columns in $E$:
\[
C(t) = \{1, 2, \ldots |e_k|-1\} \cup \{ |e_k|-1+t \}, 
\qquad 1 \le t \le d-|e_k|
\]
In other words, $C(t)$ contains the first $|e_k|-1$ columns together with column $|e_k|-1+t$.

Now
the incidence constraint for the pinning subspace $q_k$ is represented
as $m_k (d-|e_k|)$ equations
for all the $m_k$ pins $\{x^k_1, x^k_2, \ldots x^k_{m_k}\}$.
Consequently, the pinned subspace-incidence  problem reduces to solving a system 
of  $\sum_{k=1}^{|E|} m_k (d-|e_k|)$ equations, each of form~\eqref{eq:constraint}.  
We denote this algebraic system by $\algesystem = 0$.

\subsection{Linearization and Genericity}

We are interested in characterizing \emph{minimal rigidity} of pinned subspace-incidence systems.
However, checking independence relative to the ideal generated by the variety is computationally hard and best known algorithms, such as computing Gr\"{o}bner basis, are exponential in time and space~\cite{mittmann2007grobner}.
So we define the notion of rigidity for frameworks $(H,Q,p)$, which is equivalent to maximal rank of the Jacobian
of $\algesystem$. 

\begin{definition}
	A pinned subspace-incidence framework $\framework$ is  \emph{rigid} if there exists a neighborhood $N(p)$ such that  $\framework[p]$ is the only framework realizing $(H,Q)$ in $N(p)$.
	A rigid pinned subspace-incidence framework $\framework$ is \emph{minimally rigid} if it is no longer rigid after removing any pin.
\end{definition}

A \emph{generic} 
 framework with respect to a property $\mathcal{P}$, when viewed as a point in an appropriate real space, avoids a measure-zero set $\mathcal{N_\mathcal{P}}$ of the ambient space of frameworks that depends only on the underlying weighted hypergraph.
This implies the following notion of genericity:

\begin{definition}
	\label{def:generic_framework}
	A pinned subspace-incidence framework $\framework$ is \emph{generic} w.r.t.\ a property $\mathcal{P}$ if and only if 
	there exists a neighborhood $N(Q,p)$ such that for all frameworks $(H, Q', p')$ with $(Q',p') \in N(Q,p)$,
	$(H, Q', p')$ satisfies $\mathcal{P}$ if and only if $\framework$ satisfies $\mathcal{P}$.
\end{definition}

Furthermore we can define when a property is \emph{generic}, i.e.\ becomes a property of the hypergraph underlying a geometric constraint system.

\begin{definition}
	\label{def:generic_property}
	A property $\mathcal{P}$ is \emph{generic} (i.e, becomes a property of the underlying weighted hypergraph alone) if for any weighted hypergraph $H = (V,E,m)$, 
	either all generic (w.r.t.\ $\mathcal{P}$) frameworks $\framework$ satisfy $\mathcal{P}$, or all generic (w.r.t.\ $\mathcal{P}$) frameworks $\framework$ do not satisfy $\mathcal{P}$.
\end{definition}



The primary activity of the area of combinatorial rigidity is to
give purely combinatorial characterizations of generic properties $\mathcal{P}$. 
In practice, the set $\mathcal{N_\mathcal{P}}$ defining genericity is usually not specified, as long as it is of measure zero
in the ambient space of frameworks.
The measure-zero set $\mathcal{N_\mathcal{P}}$ 
may include zero-sets of some polynomials called \emph{pure conditions}
that appear in the process of drawing such combinatorial characterizations
(we will see this in the proof of Theorem~\ref{thm:rigidity_condition}).

\smallskip
We define the generic rigidity of a pinned subspace-incidence system to be the rigidity of a generic framework. 
\begin{definition}	
	A pinned subspace-incidence system $(H,Q)$ is \emph{generically (minimally) rigid} if some generic framework $\framework$ realizing $(H,Q)$ is (minimally) rigid.
\end{definition}

\subsubsection{Linearization as Rigidity Matrix}

Next we follow the approach taken by  traditional combinatorial rigidity theory~\cite{asimow1978rigidity,graver1993combinatorial} 
to show that rigidity and independence (based on nonlinear polynomials) of pinned subspace-incidence systems 
are generically properties of the underlying weighted hypergraph $H$,
and can furthermore be captured by linear conditions in an infinitesimal setting.
Specifically, we give a lemma showing that
generic rigidity of a pinned subspace-incidence system
is equivalent  to existence of a full rank  \emph{rigidity matrix}, 
obtained by taking the Jacobian of the algebraic system $\algesystem = 0$ at a generic framework $\framework$ realizing $(H,Q)$.

A \emph{rigidity matrix} of a framework $\framework$ is the  whose kernel is the infinitesimal motions (flexes) of $\framework$. 
A framework is \emph{infinitesimally rigid} if 
the rigidity matrix has full rank. 
%
To define a rigidity matrix for a pinned subspace-incidence framework $\framework$,
we take the Jacobian of the algebraic system $\algesystem = 0$
by taking partial derivatives with respect to the coordinates of $p_i$'s. 
In the Jacobian, each vertex $v_i$ has $d-1$  corresponding columns, 
and 
each pinning subspace $q_k$ associated with the hyperedge  $e_k = \{v^k_1, v^k_2, \ldots, v^k_{|e_k|} \}$
has $m_k(d-|e_k|)$ corresponding rows, 
where each equation $\det\left(E^k_l[\Cdot, C(t)]\right)=0$ \eqref{eq:constraint},
i.e.\ equation $t$ of the pin $x^k_l$,
gives the following row
(the columns corresponding to vertices not in $e_k$ are all zero): 

\fussy

{\small
	\begin{align}
	\label{eq:row}
	\hspace{-10pt}
	\bigg[\, 0,&\ldots,0, \frac{\partial \det E^k_l[\Cdot, C(t)] }{\partial p^k_{1,1}}, \frac{\partial \det E^k_l[\Cdot, C(t)] }{\partial p^k_{1,2}}, \ldots,\frac{\partial \det E^k_l[\Cdot, C(t)] }{\partial p^k_{1,d-1}}, 0,\ldots \notag \\ 
	&\ldots,0,\frac{\partial \det E^k_l[\Cdot, C(t)] }{\partial p^k_{2,1}}, \frac{\partial \det E^k_l[\Cdot, C(t)] }{\partial p^k_{2,2}},  \ldots, \frac{\partial \det E^k_l[\Cdot, C(t)] }{\partial p^k_{2,d-1}}, 0,\ldots \notag \\ 
	&\ldots \ldots \notag \\ 
	&\ldots, 0, \frac{\partial \det E^k_l[\Cdot, C(t)] }{\partial p^k_{|e_k|,1}}, 
	\frac{\partial \det E^k_l[\Cdot, C(t)] }{\partial p^k_{|e_k|,2}},  
	\ldots, \frac{\partial \det E^k_l[\Cdot, C(t)] }{\partial p^k_{|e_k|,d-1}}, 0,\ldots,0 \, \bigg	]
	\end{align}
}

Let $V^k$ be the matrix whose rows are coordinates of $p^k_1, p^k_2, \ldots, p^k_{|e_k|}$:
\[
\left[
\begin{array}{cccc}
p_{1,1} & p_{1,2} & \ldots & p_{1, d-1}\\
p_{2,1} & p_{2,2} & \ldots & p_{2, d-1}\\
\vdots & \vdots & \ddots & \vdots \\
p_{|e_k|,1} & p_{|e_k|,2} & \ldots & p_{|e_k|, d-1}
\end{array}
\right]
\]
Let $V^k_t$ be the $V^k[\Cdot, C(t)]$, i.e.\ the $|e_k| \times |e_k|$ submatrix of $V^k$ containing only columns in $C(t)$.
Let $V^k_{t,j}$ ($1 \le j \le d-1$) be the matrices obtained from $V^k_t$ by replacing the column corresponding to coordinate $j$ with the all-ones vector $(1, 1, \ldots, 1)$ for $j \in C(t)$, and the zero matrix for $j \notin C(t)$. 
Let $D^k_{t,j}$ be the determinant of $V^k_{t,j}$.
Let $b^{k,l}_i$ ($1 \le i \le |e_k|$) be the barycentric coordinates of the pin $x_l$ with respect to the points $p^k_i$,
i.e.\ $x_l  = \sum_{i=1}^{|e_k|} b^{k,l}_i p^k_i$.
%
%
Now \eqref{eq:row} can be rewritten in the following simplified form:


\begin{align}
r^k_{t,l} = \Big[\;0,\, &\ldots,\, 0,\, 0,\,   D^k_{t,1} b^{k,l}_1,\,  D^k_{t,2} b^{k,l}_1,\,  \ldots,\,  D^k_{t,d-1} b^{k,l}_1,\,  0,\, 0,\, \notag\\
&\ldots,\, 0,\, 0,\,  D^k_{t,1} b^{k,l}_2,\,  D^k_{t,2} b^{k,l}_2,\,  \ldots,\,  D^k_{t,d-1} b^{k,l}_2,\,  0,\, 0,\, \ldots \notag\\
&\ldots \ldots \notag\\
&\ldots,\,  0,\, 0,\,  D^k_{t,1} b^{k,l}_{|e_k|},\,  D^k_{t,2} b^{k,l}_{|e_k|},\,  \ldots,\,  D^k_{t,d-1} b^{k,l}_{|e_k|},\,  0,\, 0,\, \ldots,\, 0 \;\Big] \label{eq:row_pattern}
\end{align}

Each vertex $v^k_i$ has the entries $D^k_{t,j} b^{k,l}_i, 1 \le j \le d-1$ in its $d-1$ columns, among which exactly $|e_k|$ entries with $j \in C(t)$, i.e.\ the first $|e_k|-1$ columns together with column~$|e_k|-1+t$, are generically non-zero.
%
%
Note that 
the terms $D^k_{t, |e_k|-1+t}$  are equal for all $t$, 
so we may just use $D^k$ to denote it.

For each $1 \le t \le d-|e_k|$, 
there are $m_k$ rows  as \eqref{eq:row_pattern}, where each pin $x^k_l$ corresponds to the row $r^k_{t, l}$ for $1 \le l \le m_k$.
These $m_k$ rows have exactly the same row pattern 
except for different $b^{k,l}_i$'s:
\[
\def\arraystretch{1.4}
\scalemath{0.983}{
	\begin{blockarray}{*{12}{c}}
	& v_{1,1} & v_{1,2} & \ldots & v_{1,d-1} & & & v_{|e_k|,1} & v_{|e_k|,2} & \ldots & v_{|e_k|,d-1} & \\
	\begin{block}{[>{\hspace{0.1em}}*{12}{c}<{\hspace{0.1em}}]}
	\ldots & D^k_{t,1} b^{k,1}_1 & D^k_{t,2} b^{k,1}_1 & \ldots & D^k_{t,d-1} b^{k,1}_1 &
	\ldots & 
	\ldots & D^k_{t,1} b^{k,1}_{|e_k|} & D^k_{t,2} b^{k,1}_{|e_k|} & \ldots & D^k_{t,d-1} b^{k,1}_{|e_k|} &\ldots \\
	\ldots & D^k_{t,1} b^{k,2}_1 & D^k_{t,2} b^{k,2}_1 & \ldots & D^k_{t,d-1} b^{k,2}_1 &
	\ldots & 
	\ldots & D^k_{t,1} b^{k,2}_{|e_k|} & D^k_{t,2} b^{k,2}_{|e_k|} & \ldots & D^k_{t,d-1} b^{k,2}_{|e_k|} &\ldots \\
	\BAmulticolumn{12}{c}{\ddots}\\
	\ldots & D^k_{t,1} b^{k,m_k}_1 & D^k_{t,2} b^{k,m_k}_1 & \ldots & D^k_{t,d-1} b^{k,m_k}_1 &
	\ldots & 
	\ldots & D^k_{t,1} b^{k,m_k}_{|e_k|} & D^k_{t,2} b^{k,m_k}_{|e_k|} & \ldots & D^k_{t,d-1} b^{k,m_k}_{|e_k|} &\ldots \\
	\end{block}
	\end{blockarray}
}
\]

\begin{example}
	For $d=4$, consider a pinning subspace $q$ with $m(q) = 2$ associated with the hyperedge $e = \{v_1, v_2\}$.
	The pinning subspace has the following $m(q) \cdot (d - |e|) = 4$ rows in the simplified Jacobian (the index $k$ is omitted):
	\[
	\scalemath{0.983}{
		\begin{blockarray}{c@{\hspace{1.5em}}c@{\hspace{1.5em}}c@{\hspace{1.5em}}c@{\hspace{1.5em}}c@{\hspace{1.5em}}c@{\hspace{1.5em}}c@{\hspace{1.5em}}c@{\hspace{1.5em}}c@{\hspace{1.5em}}c@{\hspace{1.5em}}}
		& & v_{1,1} & v_{1,2} & v_{1,3} & & v_{2,1} & v_{2,2} & v_{2,3} & \\
		\begin{block}{c@{\hspace{1.5em}}[c@{\hspace{1.5em}}c@{\hspace{1.5em}}c@{\hspace{1.5em}}c@{\hspace{1.5em}}c@{\hspace{1.5em}}c@{\hspace{1.5em}}c@{\hspace{1.5em}}c@{\hspace{1.5em}}c@{\hspace{1.5em}}]}
		t=1, l=1 & \quad\ldots & D_{1,1} b^1_1& D b^1_1& 0 & \ldots & D_{1,1} b^1_2 & D b^1_2 & 0 & \ldots\quad \\
		t=1, l=2 & \quad\ldots & D_{1,1} b^2_1 & D b^2_1 & 0 & \ldots & D_{1,1} b^2_2 & D b^2_2 & 0 & \ldots\quad \\
		t=2, l=1 & \quad\ldots & D_{2,1} b^1_1& 0 & D b^1_1& \ldots & D_{2,1} b^1_2 & 0 & D b^1_2 & \ldots\quad \\
		t=2, l=2 & \quad\ldots & D_{2,1} b^2_1 & 0 & D b^2_1 & \ldots & D_{2,1} b^2_2 & 0 & D b^2_2 & \ldots\quad \\
		\end{block}
		\end{blockarray} 
	}
	\]
\end{example}

We define the \emph{rigidity matrix} $M\framework$ or simply $M(p)$ for a pinned subspace-incidence framework $\framework$ to be the simplified Jacobian matrix obtained above, where each row has form \eqref{eq:row_pattern}.
It is a matrix of size $\sum_k m_k (d - |e_k|)$ by $n(d-1)$.
In general, we use $M$ to denote the rigidity matrix of a generic framework $\framework$ with respect to infinitesimal rigidity.


\begin{lemma}
	\label{lem:generic}
	Infinitesimal rigidity of a generic subspace-incidence framework $\framework$ is equivalent to rigidity of $\framework$,  thus generic rigidity of the system $(H,Q)$.
\end{lemma}


\begin{proof}
	The proof of Lemma~\ref{lem:generic} follows the approach taken by traditional combinatorial rigidity~\cite{asimow1978rigidity}.
	
	First we show that if a framework $\framework$ is generic, infinitesimal rigidity implies rigidity.
	Consider the polynomial system $\algesystem$ of equations. 
	The Implicit Function Theorem states that 
	there exists a function $g$, such that $p=g(Q)$ on some open interval, if and only if the rigidity matrix $M$ has full rank. 
	Therefore, if the framework is infinitesimally rigid, the solutions to the algebraic system are isolated points (otherwise $g$ could not be explicit). 
	Since the algebraic system contains finitely many components, there are only finitely many such solution 
	and each solution is a $0$ dimensional point. This implies that the total number of solutions is finite,
	which is the definition of rigidity.
	
	To show that generic rigidity implies generic infinitesimal rigidity, we take the contrapositive: if a generic framework is not infinitesimally rigid,
	we show that there is a finite flex. 
	If $\framework$ is not infinitesimally rigid,
	then the rank $r$ of the rigidity matrix ${M}$ is less than $(d-1)|V|$. 
	Let $E^*$ be a set of edges in $H$ such that $|E^*|=r$ and the corresponding rows in $M$ are all independent.
	In ${M}[E^*, \Cdot]$, we can find $r$ independent columns. 
	Let $p^*$ be the components of $p$ corresponding to those $r$ independent columns and $p^{*\perp}$ be the remaining components.
	The $r$-by-$r$ submatrix ${M}[E^*, p^*]$, made up of the corresponding independent rows and columns, is invertible. 
	Then, by the Implicit Function Theorem, in a neighborhood of $p$ there exists a continuous and differentiable function $g$ such that $p^*=g(p^{*\perp})$.
	This identifies $p'$, whose components are $p^*$ and the level set of $g$ corresponding to $p^*$, such that $(H,Q)(p')=0$. The level set
	defines the finite flexing of the framework. Therefore the system is not rigid.
\end{proof}

\section{Combinatorial Rigidity Characterization}

\subsection{Required Hypergraph Properties}
\label{sec:graph}

This section introduces pure hypergraph properties and definitions that will be used in stating and proving our main theorem.

\begin{definition} \label{def:tightness}
	A hypergraph $H=(V,E)$ is $(k,0)$-sparse if for any ${V'} \subset V$, 
	the induced subgraph ${H'}=({V'},{E'})$ satisfies $|{E'}| \leq k|{V'}|$. 
	A hypergraph $H$ is $(k,0)$-tight if $H$ is $(k,0)$-sparse and $|E| = k|V|$.
\end{definition}

This is a special case of the $(k,l)$-sparsity condition 
that was widely studied in the geometric constraint solving and combinatorial rigidity literature.
A relevant concept from graph matroids is \emph{map-graph}, defined as follows.

\begin{definition}
	An \emph{orientation} of a hypergraph is given by identifying as the \emph{tail} of each edge one of its endpoints.
	The \emph{out-degree} of a vertex is the number of edges which identify it as the tail and connect $v$ to $V - v$.
	A \emph{map-graph} is a hypergraph that admits an orientation such that the out degree of every vertex is
	exactly one. 
\end{definition}

The following lemma from~\cite{streinu2009sparse} follows Tutte-Nash Williams~\cite{nash1961edge,tutte1961problem}
to give a useful characterization of $(k,0)$-tight graphs in terms of maps.

\begin{lemma} 
	\label{lem:map_decomposition}
	A hypergraph $H$ is composed of $k$ edge-disjoint map-graphs if and only if $H$ is $(k,0)$-tight.
\end{lemma}

Our characterization of rigidity of a weighted hypergraph $H$ is based on map-decomposition of
a \emph{multi-hypergraph} $\widehat{H}$ obtained from $H$.

\begin{definition}
	Given a weighted hypergraph $H = (V, E, m)$,
	the associated \emph{multi-hypergraph} $\widehat{H} = (V, \widehat{E})$ is obtained by 
	replacing each hyperedge $e_k$ in $E$ with a set ${E}^k$ of $m_k (d-|e_k|)$ copies of \emph{multi-hyperedges}.
\end{definition}

A \emph{labeling} of a multi-hypergraph $\widehat{H}$ gives a one-to-one correspondence between $E^k$ and 
the set $R^k$ of $m_k (d-|e_k|)$ rows for the hyperedge $e_k$ in the rigidity matrix $M$,
where the multi-hyperedge corresponding to the row $r^k_{t,l}$ is labeled $e^k_{t,l}$.


\smallskip
\noindent\textbf{Note:} 
an alternative representation commonly adopted in geometric constraint solving~\cite{ait1993reduction,fudos1997graph,hoffmann1997symbolic} is 
to represent $H$ as a bipartite graph $B(H)$, with $d-1$ copies of each vertex in $V$ as one of its vertex set, and $m_k (d-|e_k|)$ copies of each hyperedge in $E$ as the other vertex set.
A combinatorial rigidity characterization can be equivalently stated using either flow based conditions on $B(H)$
or hypergraph sparsity conditions on $\widehat{H}$~\cite{hoffmann1997finding,hoffmann1998geometric,jacobs1997algorithm}.
However, such an equivalence of the two combinatorial properties has no bearing on  the proof of equivalence of either combinatorial property to generic rigidity, an algebraic property.  
Showing that the combinatorial property generically implies the algebraic property is  the substance of the proof of the theorem.
This is generally called the ``Laman direction''  and it is in fact where
the hardness of every combinatorial characterization of rigidity lies.

\subsection{Characterizing Rigidity}
\label{sec:proof}

In this section, we apply~\cite{white1987algebraic} to give combinatorial characterization for minimal rigidity of pinned subspace-incidence systems.

\begin{theorem}[main theorem] \label{thm:rigidity_condition}
	A pinned subspace-incidence system with pins being in general position is generically minimally rigid 
	if and only if:
	\begin{enumerate}

		\item [(1)]
		The underlying weighted hypergraph $H = (V,E, m)$ satisfies
		$\sum_{k=1}^{|E|} m_k(d-|e_k|) = (d-1)|V|$, 
		and $\sum_{e_k \in E'} m_k(d-|e_k|) \le (d-1)|V'|$ for every vertex induced subgraph $H'=(V',E')$.
		In other words, the associated multi-hypergraph $\widehat{H} = (V, \widehat{E})$ has a decomposition into $(d-1)$ maps. 
		
		\item [(2)] There exists a labeling of $\widehat{H}$ \emph{compatible with the map-decomposition (defined later)}  such that in each set $E^k$ of multi-hyperedges,
		
		(2a) two multi-hyperedges $e^k_{t_1,l_1}$ and $e^k_{t_2,l_2}$ with $l_1 = l_2$
		are not contained in the same map in the map-decomposition, 
		
		(2b) two multi-hyperedges $e^k_{t_1,l_1}$ and $e^{k}_{t_2,l_2}$ with $ t_1=t_2$ do not have the same vertex as tail in the map-decomposition. 
		
	\end{enumerate}
	
\end{theorem}


To prove Theorem~\ref{thm:rigidity_condition}, we apply Laplace expansion to the determinant of the rigidity matrix $M$, 
which corresponds to decomposing the $(d-1,0)$-tight multi-hypergraph $\widehat{H}$ as a union of $d-1$ maps.
We then prove $\det (M)$ is not identically zero by showing that the minors corresponding to each map are not identically zero, as long as a certain polynomial called \emph{pure condition} is avoided by the framework. 

A Laplace expansion rewrites the determinant of the rigidity matrix $M$ 
as a sum of products of determinants (brackets) representing each of the coordinates taken separately. 
In order to see the relationship between the Laplace expansion and the map-decomposition, we first group the columns of $M$ into $d-1$ column groups $C_j$ according to the coordinates, 
where columns for the first coordinate of each vertex belong to $C_1$, 
columns for the second coordinate of each vertex belong to $C_2$, etc. 


\begin{example} 
	For $d=4$, consider a pinning subspace $q$ with $m(q)=2$ associated with the hyperedge $e = \{v_1, v_2\}$.
	The regrouped rigidity matrix has $d-1 = 3$ column groups, where 
	$q$ has the following $4$ rows (the index $k$ is omitted):
	\newlength{\mylength}
	\setlength{\mylength}{0.57em}
	\newlength{\mylengtha}
	\setlength{\mylengtha}{0.3em}
	\[
	\def\arraystretch{1.2}
	\scalemath{0.983}{
		\begin{blockarray}{c@{\hspace{0.8em}}c@{\hspace{\mylength}}c@{\hspace{\mylength}}c@{\hspace{\mylength}}c@{\hspace{\mylength}}c@{\hspace{\mylength}}c@{\hspace{\mylength}}c@{\hspace{\mylength}}c@{\hspace{\mylength}}c@{\hspace{\mylength}}c@{\hspace{\mylength}}c@{\hspace{\mylength}}c@{\hspace{\mylength}}c@{\hspace{\mylength}}c@{\hspace{\mylength}}c@{\hspace{\mylength}}}
		&  & v_{1,1} & & v_{2,1} & & & v_{1,2} && v_{2,2} & & & v_{1,3} && v_{2,3} & \\
		\begin{block}{c@{\hspace{0.8em}}[c@{\hspace{\mylength}}c@{\hspace{\mylength}}c@{\hspace{\mylength}}c@{\hspace{\mylength}}c@{\hspace{\mylengtha}}|@{\hspace{\mylengtha}}c@{\hspace{\mylength}}c@{\hspace{\mylength}}c@{\hspace{\mylength}}c@{\hspace{\mylength}}c@{\hspace{\mylengtha}}|@{\hspace{\mylengtha}}c@{\hspace{\mylength}}c@{\hspace{\mylength}}c@{\hspace{\mylength}}c@{\hspace{\mylength}}c@{\hspace{\mylength}}]}
		t=1,l=1 & \;\ldots & D_{1,1}b^1_1 & \ldots & D_{1,1}b^1_2 & \ldots & \ldots & Db^1_1& \ldots & Db^1_2 & \ldots & &&&& \\
		t=1,l=2 & \;\ldots & D_{1,1} b^2_1 & \ldots & D_{1,1} b^2_2 & \ldots & \ldots & D b^2_1 & \ldots & D b^2_2 & \ldots & &&&& \\
		t=2,l=1 & \;\ldots & D_{2,1} b^1_1& \ldots & D_{2,1} b^1_2 & \ldots & &&&& & \ldots & D b^1_1& \ldots  & D b^1_2 & \ldots \; \\
		t=2,l=2 & \;\ldots & D_{2,1} b^2_1 & \ldots & D_{2,1} b^2_2 & \ldots & &&&& & \ldots & D b^2_1 & \ldots  & D b^2_2 & \ldots \; \\
		\end{block}
		\end{blockarray} 
	}
	\]
\end{example}

We have the following observation on the pattern of the regrouped rigidity matrix.
\begin{observation}
	\label{obs:mat_pattern}
	In the rigidity matrix $M$ with columns grouped into column groups,
	a hyperedge $e_k$ has $m_k(d-|e_k|)$ rows, each associated with a  multi-hyperedge of $e_k$ in $\widehat{H}$. 
	In a column group $j$ where $j \le |e_k|-1$, each row associated with $e_k$
	contains $|e_k|$ nonzero entries at the columns corresponding to vertices of $e_k$.
	In a column group $j$ where $j \ge |e_k|$, 
	there are $m_k$ rows $r^k_{t,l}$ with $|e_k| - 1 + t = j$, each containing $|e_k|$ nonzero entries at the columns corresponding to vertices of $e_k$,
	while  the remaining rows  are all zero.
\end{observation}

A \emph{labeling of $\widehat{H}$ compatible with a given map-decomposition} can be obtained as following.
We start from the last column group of $M$ and associate each column group $j$ with a map in the map-decomposition.
For each multi-hyperedge of the map that is a copy of the hyperedge $e_k$, 
we pick a row $r^k_{t,j}$ that is not all zero in column groups $j$ and label the multi-hyperedge as $e^k_{t,j}$.
By the above observation, this is always possible if each map contains at most $m_k$ multi-hyperedges of the same hyperedge $e_k$, which must be true if there exists any labeling of $\widehat{H}$ satisfying  Theorem~\ref{thm:rigidity_condition}(2a).

In the Laplace expansion
\begin{equation}
\det({M}) = \sum_{\sigma} \left( \pm \prod_j \det {M}[R_j^\sigma, C_j] \right) \label{eq:laplace}
\end{equation}
the sum is taken over all partitions $\sigma$ of the rows into $d-1$ subsets
$R^\sigma_1, R^\sigma_2,$ $\ldots,$ $R^\sigma_j,$ $\ldots,$ $R^\sigma_{d-1}$, each of size $|V|$.
In other words, each summation term of~\eqref{eq:laplace} contains $|V|$ rows $R_j^\sigma$ from each column group $C_j$.
Observe that for any submatrix ${M}[R^\sigma_{j}, C_{j}]$, each row has a common coefficient $D^k_{t,j}$, so
\[
\det({M}[R^\sigma_{j}, C_{j}]) = \left(\prod_{r^k_{t,j} \in R^\sigma_j} D^k_{t,j} \right) \det(M'[R^\sigma_{j}, C_{j}])
\]
where each row of $M'[R^\sigma_{j}, C_{j}]$ is either all zero, or of the pattern
\begin{equation}
[0,\ldots, 0, b^{k,l}_1, b^{k,l}_2, 0, \ldots ,  b^{k,l}_{|e_k|}, 0, \ldots, 0] \label{eq:submat_rowpattern}
\end{equation}
with non-zero entries only at the $|e_k|$ indices corresponding to $v^k_i \in e_k$.

For a fixed $\sigma$, we refer to a submatrix  ${M}[R^\sigma_{j}, C_{j}]$ simply as $M_j$.

\begin{example}
	Figure~\ref{fig:d3_rigid} shows a pinned subspace-incidence system in $d=3$ with 4 vertices and 5 hyperedges, 
	where $e_1 = \{v_1\}, e_2 = \{v_2\}, e_3 = \{v_1, v_3\}, e_4 = \{v_2, v_4\}, e_5=\{v_3,v_4\}$, 
	and $m_k = 1$ for all $1 \le k \le 5$ except that $m_5 = 2$.
	Figure~\ref{fig:d3_rigid_map} gives a map-decomposition of the multi-hypergraph $\widehat{H}$ of (a).
	The labeling of multi-hyperedges is given in the  regrouped rigidity matrix  \eqref{eq:map_decomp}, 
	where the shaded rows inside the column groups constitute the submatrices $M^\sigma_j$ in the summation term of the Laplace decomposition corresponding to the map-decomposition.
	The system is generically minimally rigid, as the map-decomposition and labeling of $\widehat{H}$ satisfies the conditions of Theorem~\ref{thm:rigidity_condition}.
	\begin{figure}
		\centering
		\begin{subfigure}{.55\linewidth}
			\centering
			\includegraphics[width=.95\linewidth]{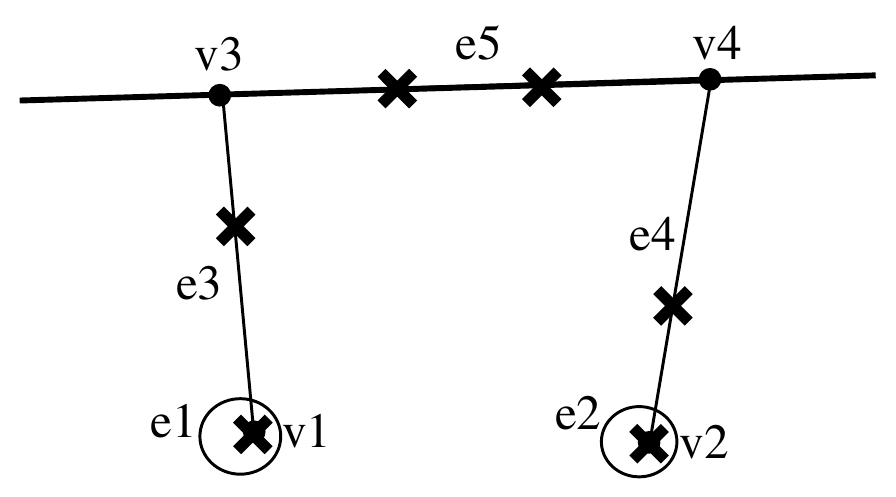}
			\caption{}
			\label{fig:d3_rigid}
		\end{subfigure}%
		\begin{subfigure}{.45\linewidth}
			\centering
			\includegraphics[width=.8\linewidth]{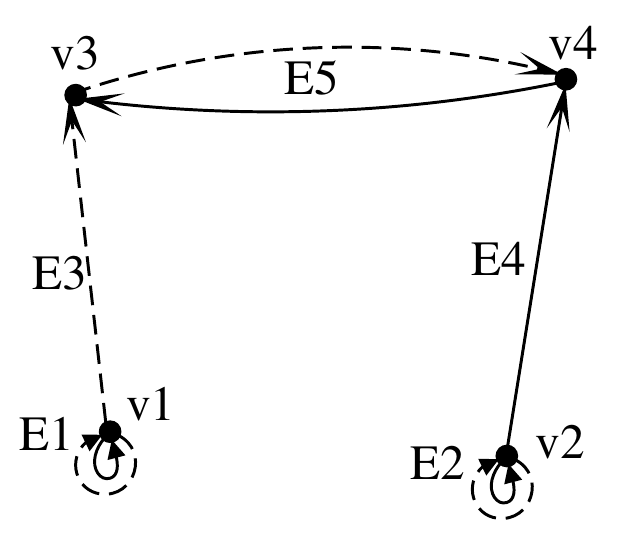}
			\caption{}
			\label{fig:d3_rigid_map}
		\end{subfigure}
		\caption{(a) A minimally rigid pinned subspace-incidence system in $d=3$.  
			(b) A map-decomposition of the multi-hypergraph of the system in (a), where differently drawn multi-hyperedges are in different maps, and the tail vertex of each multi-hyperedge is pointed to by an arrow.
		}
	\end{figure}

	\begin{equation*}
		\includegraphics[width=.95\linewidth]{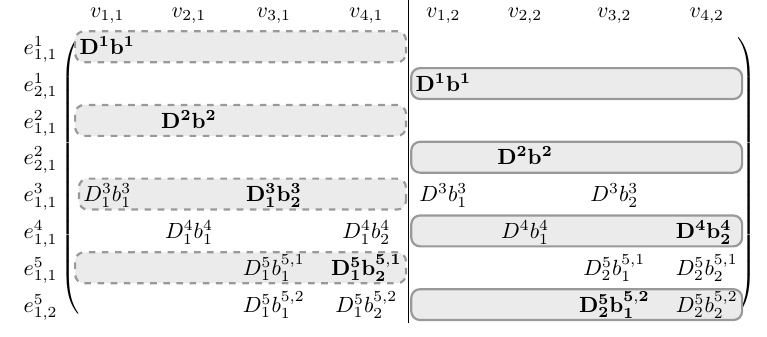}
	\addtocounter{equation}{1}\tag{\theequation} \label{eq:map_decomp}
	\end{equation*}
\end{example}

\begin{proof}[Proof of main theorem]
	First we show the only if direction. 
	For a generically minimally rigid pinned subspace-incidence framework, the rigidity matrix $M$ is generically full rank,
	so there exists at least one summation term $\sigma$ in \eqref{eq:laplace} where each 
	submatrix ${M}_j$ is generically full rank.
	As the submatrices don't have any overlapping rows with each other,
	we can perform row elimination on $M$ to obtain a matrix $N$ with the same rank, 
	where all submatrices ${M}_j$ are simultaneously 
	converted to a \emph{permuted reduced row echelon form} ${N}_j$,
	where each row in $N_j$ has exactly one non-zero entry $\beta^j_i$ at a unique column $i$, 
	In other words, all ${N}_j$'s can be converted simultaneously to reduced row echelon form
	by multiplying a permutation matrix on the left of $N$.
	%
	Now we can obtain a map-decomposition of $\widehat{H}$
	by letting each map $j$ contain multi-hyperedges corresponding to rows of the submatrix $N_j$,
	and assigning each multi-hyperedge in map $j$  the vertex $i$ corresponding to the non-zero entry $\beta^j_i$ in the associated row in ${N}_j$ as tail.
	In addition, such a map-decomposition must satisfy the conditions of Theorem~\ref{thm:rigidity_condition}(2):
	
	Condition~(2a): assume two multi-hyperedges $e^k_{t_1,l}$ and $e^k_{t_2,l}$ are in the same map $j$,
	i.e.\ the rows corresponding to these two edges are included in the same submatrix $M_j$. 
	If $j > s-1$, one of these rows must be all-zero in $M_j$ by Observation~\ref{obs:mat_pattern}, contradicting the condition that $M_j$ is full rank. 
	If $j \le s-1$, both of these rows in $M_j$ will be a multiple of the same row vector \eqref{eq:submat_rowpattern},
	contradicting the condition that $M_j$ is full rank. 

	Condition~(2b): 
	note that the rows  in $M$ corresponding to multi-hyperedges $e^k_{t,l_1}$ and $e^k_{t,l_2}$ have the exactly the same pattern except for different values of $b$'s.
	If $e^k_{t,l_1}$ has vertex $i$ as  tail, after the row elimination, the column containing $b^{k,l_1}_{i}$ will become the only non-zero entry of column $i$ for $N_{j_1}$, 
	while the column containing $b^{k,l_2}_{i}$ will become zero in all column groups, 
	thus $i$ cannot be assigned as tail for $e^k_{t,l_2}$.
	
	

	\medskip
	
	Next we show the if direction, that the conditions of Theorem~\ref{thm:rigidity_condition} imply infinitesimal rigidity. 
	
	Given labeled multi-hypergraph $\widehat{H}$ with a map-decomposition satisfying the conditions of Theorem~\ref{thm:rigidity_condition}, 
	we can obtain  summation term $\sigma$ in the Laplace decomposition~\eqref{eq:laplace} 
	according to the labeling of $\widehat{H}$, where
	%
	each submatrix $M_j$ contain all rows corresponding to the map associated with column group $j$.

	First, it is not hard to show that each submatrix $M_j$ is generically full rank~\cite{whiteley1989matroid}. 
	For completeness, we give a short proof as following.
	According to the definition of a map-graph, 
	the function $\tau:\widehat{E} \rightarrow V$ assigning a tail vertex to each multi-hyperedge  is a one-to-one correspondence.
	We perform symbolic row elimination of the matrix $M$ to simultaneously convert each $M_j$ to its permuted reduced row echelon form $N_j$, where for each row of $N_j$, all entries are zero except for the entry $\beta^k_{t,l}$ corresponding to the vertex $\tau(e^k_{t,l})$, which is a polynomial in  $b^{k,l}_i$'s in the submatrix $M_j$.
	Since $M_j$ cannot contain two rows with the same $k$ and $l$ by Condition (2a), the $b^{k,l}_i$'s in different rows of a same map are independent of each other, $\beta^k_{t,l} \ne 0$ under a generic specialization of $b^{k,l}_i$.
	%
	Since each row of $N_j$ has exactly one nonzero entry and the nonzero entries from different rows are on different columns, the $|V| \times |V|$ matrix $N_j$ is clearly full rank. Thus $M_j$ must also be generically full rank.
	
	We conclude that 
	\begin{equation}
	\det({M}) = \sum_{\sigma} \left( \pm \prod_j \Bigg( \Big(\prod_{r^k_{t,j} \in R^\sigma_j} D^k_{t,j} \Big) \det {M}'[R_j^\sigma, C_j] \Bigg) \right) \label{eq:pure_cond}
	\end{equation}
	where the sum is taken over all $\sigma$ corresponding to a map-decomposition of $\widehat{H}$.
	Generically, the summation terms of the sum \eqref{eq:pure_cond} do not cancel with each other, 
	since $\det(M'[R_j^\sigma,C_j])$ are independent of the multi-linear coefficients $\prod_{r^k_{t,j} \in R^\sigma_j} D^k_{t,j}$ because of
	the requirement that the pins are in general position, 
	and any two rows of $M$ are independent by Condition~(2b).
	This implies that 
	$\widehat{M}$ is generically full rank.
	
	The polynomial \eqref{eq:pure_cond} 	
	gives the pure condition characterizing non-generic frameworks.	
\end{proof}

\subsubsection{Pure condition}

The pure condition \eqref{eq:pure_cond} obtained in the proof of Theorem~\ref{thm:rigidity_condition} 
vanishes at a measure-zero subset of frameworks.
For those frameworks, the above characterization fails. 
However, the geometric meaning of the pure condition is not completely clear. 



Note that there exist combinatorial types of underlying weighted hypergraphs,
for example,  $H$ with a subgraph $(V', E')$, $|V'| < d$ such that $\sum_{e_k \in E'} m_k > |V'|$,
that force the pins to lie in a non-generic position no matter how the vertices are realized.
We rule out such systems by the requirement that pins being in general position in the statement of Theorem~\ref{thm:rigidity_condition}.

	As an example, Figure~\ref{fig:counter_sys} shows a  pinned subspace-incidence system in $d=4$ with 4 vertices and 5 hyperedges, where 
	$m_k$ is $2$ for $k=5$ and is $1$ otherwise. 
	A map-decomposition of the multi-hypergraph $\widehat{H}$ of the system is given in Figure~\ref{fig:counter_map}, and we can easily find a labeling of $\widehat{H}$ satisfying Conditions (1) and (2) in Theorem~\ref{thm:rigidity_condition}.
	However, the system is overconstrained, as generically the four pins associated with $e_1, e_4$ and $e_5$ will not fall on the same plane in $\proj[3]$.
	This situation is both captured  by the pure condition and is ruled by the requirement of Theorem~\ref{thm:rigidity_condition} that the pins being in general position.
	
	
	
	\begin{figure}
		\centering
		\begin{subfigure}{.4\linewidth}
			\centering
			\includegraphics[width=.95\linewidth]{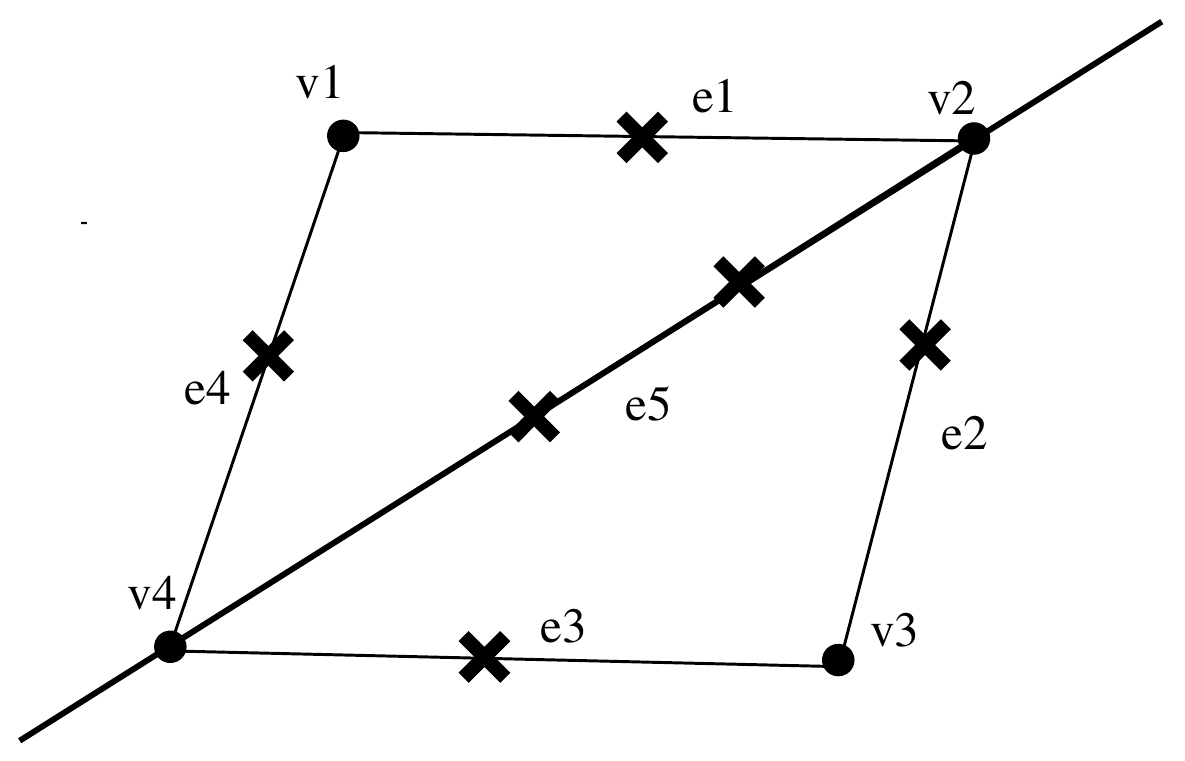}
			\caption{}
			\label{fig:counter_sys}
		\end{subfigure}%
		\begin{subfigure}{.4\linewidth}
			\centering
			\includegraphics[width=.8\linewidth]{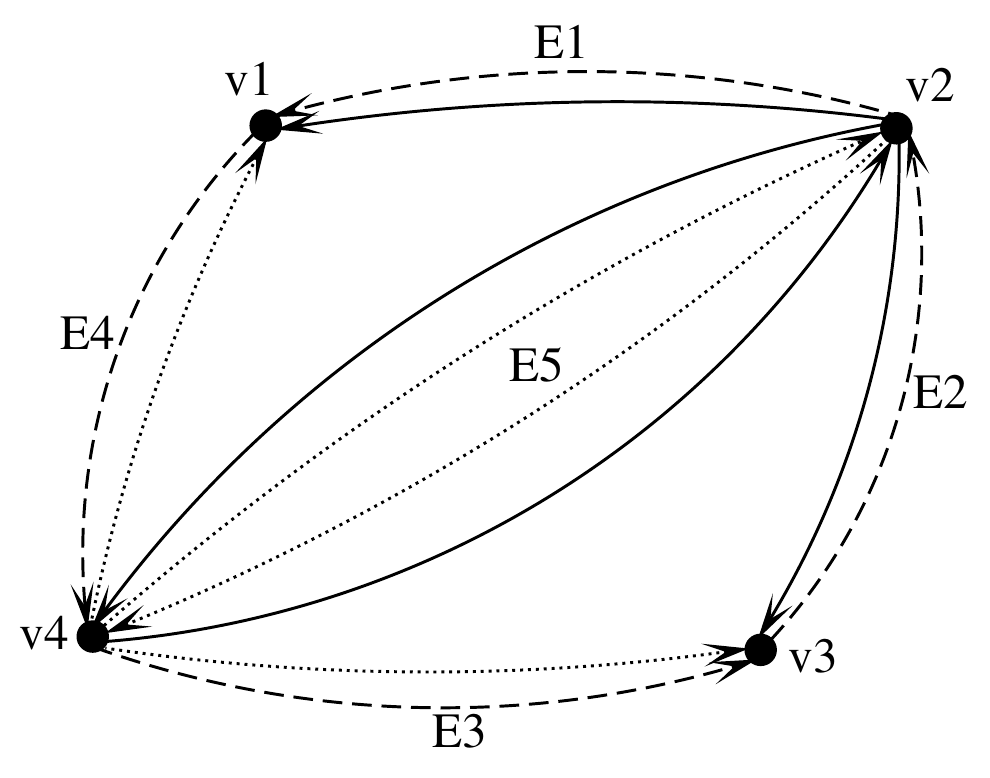}
			\caption{}
			\label{fig:counter_map}
		\end{subfigure}
		\caption{(a) A pinned subspace-incidence system in $d=4$.  
			(b) A map-decomposition of the multi-hypergraph of the system in (a), where multi-hyperedges with different patterns are in different maps, and the tail vertex of each multi-hyperedge is pointed to by an arrow.
		}
	\end{figure}
	
%
	

Figure~\ref{fig:pure} shows a more standard non-generic example captured by the pure condition. Frameworks (a) and (b) are two pinned subspace incidence frameworks with $d=4$, and they have the same underlying weighted hypergraph satisfying the combinatorial characterization of the main theorem. However, framework (a) is minimally rigid but (b) is not, since the pins on hyperedges $e_1$ and $e_2$ in (b) lie on the same line. Evaluating the pure condition at framework (b) shows that it is not generic.

	\begin{figure}
		\centering
			\centering
			\includegraphics[width=.6\linewidth]{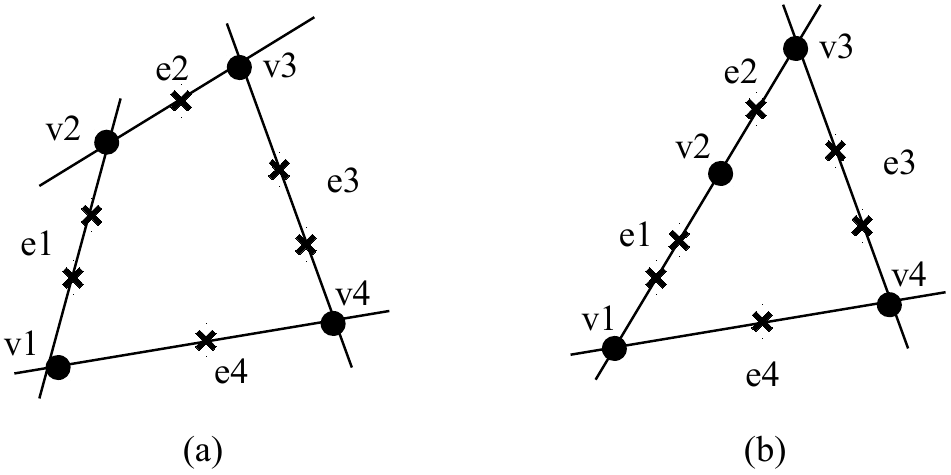}
		\caption{ Two  pinned subspace-incidence  frameworks in $d=4$ with the same underlying weighted hypergraph, where (a) is generic but (b) is not generic. }
		\label{fig:pure}
	\end{figure}

%
%
%

\section{Application: Dictionary Learning}

Dictionary Learning (aka Sparse Coding) is the problem of obtaining a sparse representation
of data points, by learning \emph{dictionary vectors} upon which the data points
can be written as sparse linear combinations.
The Dictionary Learning problem arises in various context(s) such as signal processing and machine learning.
\begin{problem}[Dictionary Learning] \label{prob:dictionary_learning}
A point set $X=[x_1 \ldots x_m]$ in $\mathbb{R}^d$ is said to be \emph{$s$-represented} by a \emph{dictionary} $D = [v_1 \ldots v_n]$ for a given \emph{sparsity} $s < d$,
if there exists $\Theta = [\theta_1 \ldots \theta_m]$ such that $x_i=D\theta_i$,  
with $\Arrowvert\theta_i\Arrowvert_0 \leq s$ (i.e.\ $\theta_i$ has at most $s$ non-zero entries).
Given an $X$ known to be $s$-represented by an unknown \emph{dictionary} $D$ of size $|D|=n$, 
\emph{Dictionary Learning} is the problem of finding any  dictionary $\acute{D}$ 
satisfying the properties of $D$, i.e.\ $|\acute{D}| \le n$, 
and there  exists $\acute{\theta_i} $ such that $x_i=\acute{D} \acute{\theta_i}$  for all $x_i \in X$.
%
\end{problem}

The dictionary under consideration is usually \emph{overcomplete},  with $n > d$.
However we are interested in asymptotic performance with respect to all four variables $n,m,d,s$.
Typically, $m \gg n \gg d > s$. 
Both cases when $s$ is large relative to $d$ and when $s$ is small relative to $d$ are interesting.

We understand the Dictionary Learning problem from an intrinsically geometric point of view. 
Notice that each $x\in X$ lies in an $s$-dimensional subspace $\suppx{x}$,
which is the span of $s$ vectors $v \in D$ that form the \emph{support} of $x$. 
The resulting \emph{$s$-subspace arrangement} $\sxd = \{ \suppx{x}: x\in X \}$ has an underlying 
labeled (multi)hypergraph
$\hsxd = (\cald, \calsxd)$, where $\cald$ denotes the index set of the dictionary $D$ and 
$\calsxd$ is the set of (multi)hyperedges over the indices $\cald$ 
corresponding to the subspaces $\suppx{x}$.
The word ``multi'' appears because if $\suppx{x_1} = \suppx{x_2}$ for data points 
$x_1,x_2 \in X$ with $x_1\not = x_2$, 
then that subspace is multiply represented in $\sxd$ (resp.\ $\calsxd$)
as $\suppx{x_1}$ and $\suppx{x_2}$.

Note that there could be many dictionaries $D$ 
and for each $D$, many possible subspace arrangements $\sxd$ that are solutions to the 
Dictionary Learning problem. 

 \subsection{Systematic Classification of Problems Closely Related to Dictionary Learning and Previous Approaches}
 \label{sec:review}

A closely related problem to Dictionary Learning is
the Vector Selection (aka sparse recovery) problem, which finds a representation of input data in a known dictionary $D$. 

\begin{problem}[Vector Selection]
	Given a dictionary $D \in \mathbb{R}^{d \times n}$  and an input data point
	$x \in \mathbb{R}^d$, the \emph{Vector Selection problem} asks for $\theta \in \mathbb{R}^n$ 
	such that $x = D\theta$ with $\Arrowvert \theta \Arrowvert_{0}$ minimized.
\end{problem} 

That is, $\theta$ is a sparsest support vector that represents $x$ as linear combinations of the
columns of $D$.

An optimization version of Dictionary Learning can be written as:
\[
\min_{D \in \mathbb{R}^{d\times n}}{ \max_{x_i} {\min \Arrowvert \theta_i \Arrowvert_0  : x_i = D\theta_i}}.  
\]
In practice, it is often relaxed to the Lagrangian 
$ \min \sum_{i=0}^m (\Arrowvert x_i - D\theta_i \Arrowvert_2  + \lambda \Arrowvert \theta_i \Arrowvert_1$).

Several traditional Dictionary Learning algorithms work by \emph{alternating minimization}, 
i.e.\ iterating the following two steps \cite{sprechmann2010dictionary,ramirez2010classification,olshausen1997sparse}:

1. Starting from an initial estimation of $D$, solving the Vector Selection problem for all data points $X$ to find a corresponding $\Theta$. This can be done using any vector selection algorithm, such as basis pursuit from \cite{chen1998atomic}. 

2. Given $\Theta$, updating the dictionary estimation by solving the optimization problem is now convex in $D$. 
%
For an 
overcomplete dictionary, 
the general Vector Selection problem is ill defined, 
as there can be multiple solutions for a data point $x$.
Overcoming this by framing the problem as a minimization problem is exceedingly difficult.
Indeed under generic assumptions, 
the Vector Selection problem has  been shown to be NP-hard
by reduction to the Exact Cover by 3-set problem \cite{rey1994adaptive}. 
%
%
One is then tempted to conclude that Dictionary Learning is also NP-hard. 
However, this cannot be directly deduced in general, 
since even though adding a witness $D$
turns the problem into an NP-hard problem, 
it is possible that the Dictionary Learning solves to produce
a different dictionary $\acute{D}$. 

On the other hand, if  $D$ satisfies the condition of being a \emph{frame}, 
i.e. for all $\theta$ such that $\Arrowvert \theta \Arrowvert_0 \leq s$, there exists a $\delta_s$ such that
$ (1-\delta_{s}) \leq \frac{ \Arrowvert D\theta \Arrowvert_2^2}{ \Arrowvert \theta \Arrowvert_2^2} \leq (1+\delta_{s})$,
it is guaranteed that the sparsest solution to 
the Vector Selection problem can be found via $L_1$ minimization \cite{donoho2006compressed,candes2006robust}.

\sloppy

One popular alternating minimization method is the Method of Optimal Dictionary (MOD) \cite{engan1999method},
which follows a two step iterative approach using a maximum likelihood formalism, 
and uses the pseudoinverse to compute $D$: $D^{(i+1)}=X\Theta^{(i)^T} (\Theta^n \Theta^{i^T})^{-1}.$
The MOD can be extended to Maximum A-Posteriori probability setting with different priors to take into account preferences in the recovered dictionary. 

\fussy

Similarly, $k$-SVD \cite{aharon2006svd} uses a two step iterative process, with a Truncated Singular Value Decomposition to update $D$.
This is done by taking every atom in $D$ and applying SVD to $X$ and $\Theta$ restricted to only the columns that have contribution from that atom. When $D$ is restricted to be of the form $D = [ B_1, B_2 \ldots B_L ]$ where $B_i$'s are orthonormal matrices, a more efficient pursuit algorithm
is obtained for the sparse coding stage using a block coordinate relaxation. 

Though alternating minimization methods work well in practice, there is no theoretical guarantee that the their results will converge to a true dictionary.
Several recent works 
give provable algorithms under stronger constraints on $X$ and $D$.
Spielman et.\ al \cite{spielman2013exact} give an $L_1$ minimization based approach which is provable to find the exact dictionary $D$, but requires $D$ to be a basis. 
Arora et.\ al \cite{arora2013new} and Agarwal et.\ al \cite{agarwal2013exact} 
independently give provable non-iterative algorithms for learning approximation of overcomplete dictionaries. 
Both of their methods are based on an overlapping clustering approach to find data points sharing a dictionary vector, 
and then estimate the dictionary vectors from the clusters via SVD. 
The approximate dictionary found using these two algorithms can be in turn used in iterative methods like k-SVD as the initial estimation of dictionary, leading to provable convergence rate \cite{agarwal2013learning}.
However, these overlapping clustering based methods require the dictionaries to have the \emph{pairwise incoherence} property
which is much stronger than the frame property.

%
%
%

\smallskip

 By imposing a systematic series of increasingly stringent constraints on the input, 
 we classify previous approaches to Dictionary Learning 
 as well as a whole set of independently interesting problems closely related to Dictionary Learning. 
 A summary of the input conditions and results of these different types of Dictionary Learning approaches 
 can be found in Table \ref{tab:classification}. 
 
 

 
 

 A natural restriction of the general Dictionary Learning problem is the following.
 We say that a set of data points $X$ \emph{lies on} a set $S$ of $s$-dimensional subspaces 
 if  for all $x_i \in X$, there exists $S_i \in S$ such
 that $x_i \in S_i$.
 \begin{problem} [Subspace Arrangement Learning] \label{prob:subspace_arrangement_learning}
 	Let $X$ be a given set of data points that are known to lie on a set $S$ of
 	$s$-dimensional subspaces of $\mathbb{R}^d$, where $|S|$ is at most $k$.
 	(Optionally assume that the subspaces in $S$ have bases such that their
 	union is a frame).
 	\emph{Subspace arrangement learning} finds any subspace arrangement $\acute{S}$ of
 	$s$-dimensional subspaces of $\mathbb{R}^d$ satisfying these conditions, 
 	i.e.\ $|\acute{S}| \leq k$, $X$ lies on $\acute{S}$, (and optionally the union of the bases of $\acute{S}_i \in \acute{S}$ is a frame).
 \end{problem}
 
 There are several known algorithms for learning subspace arrangements. 
 Random Sample Consensus (RANSAC) \cite{vidal2011subspace} is an approach to learning subspace arrangements that isolates, one subspace at a time, via random sampling.
 When dealing with an arrangement of $k$ $s$-dimensional subspaces, for instance, the method samples $s+1$ points
 which is the minimum number of points required to fit an $s$-dimensional subspace. The procedure then finds and discards
 inliers by computing the residual to each data point relative to the subspace and selecting the points whose residual is below 
 a certain threshold. The process is iterated until we have $k$ subspaces or all points are fitted.
 RANSAC is robust to models corrupted with outliers. 
 Another method called Generalized PCA (GPCA) \cite{vidal2005generalized} 
 uses techniques from algebraic geometry   for subspace clustering, 
 finding a union of $k$ subspaces by factoring a homogeneous polynomial of
 degree $k$ that is fitted to the points $\{ x_1 \ldots x_m \}$.
 Each factor of the polynomail represents the normal vector to a subspace. 
 We note that GPCA can also determine $k$ if it is unknown.

 The next problem is obtaining a minimally sized dictionary from a subspace arrangement.
 %
 \begin{problem} [\begin{small}Smallest Spanning Set for Subspace Arrangement\end{small}] \label{prob:smallest_spanning_set}
 	Let $S$ be a given
 	set of $s$-dimensional subspaces of $\mathbb{R}^d$ specified by giving their bases. 
 	Assume their intersections are known to be $s$-represented by a set $I$ of vectors with $\arrowvert I \arrowvert$ at most $n$. 
 	Find any set of vectors $\acute{I}$ that satisfies these conditions.
 \end{problem}
 
 The smallest spanning set is not necessarily unique in general, 
 and is closely related to the intersection semilattice of subspace arrrangement \cite{bjorner1994subspace,goresky1988stratified}.
 Furthermore, 
 under the condition that the subspace arrangement comes from 
 a frame dictionary, the smallest spanning set 
 is the union of:
 (a) the smallest spanning set $I$ of the pairwise intersection of all the subspaces in $S$;
 (b) any points outside the pairwise intersections that, together with $I$, completely $s$-span the subspaces in $S$. 
 This directly leads to a recursive algorithm for the smallest spanning set problem. 

The \emph{fitted dictionary learning} problem is the version of dictionary learning where the underlying hypergraph $\hsxd$ is specified.
%
\begin{problem}[Fitted Dictionary Learning]  \label{prob:fitted_dictionary_learning}
	Let $X$ be a given set of data points in $\mathbb{R}^d$. 
	For an unknown  dictionary
	$D = [v_1,\ldots, v_n]$ that $s$-represents $X$, 
	we are given the hypergraph $\hsxd$ of the underlying subspace arrangement $\sxd$.
	Find any dictionary $\acute{D}$ of size $|\acute{D}| \leq n$, 
	that is consistent with the hypergraph $\hsxd$. 
\end{problem}

 %
 When $X$ contains sufficiently dense data to solve Problem \ref{prob:subspace_arrangement_learning}, 
 Dictionary Learning reduces to problem \ref{prob:smallest_spanning_set}, 
 i.e.\ we can use the following two-step procedure to solve the Dictionary Learning problem: 
 \begin{itemize}
 	\item Learn a Subspace Arrangement $S$ for $X$ (instance of Problem \ref{prob:smallest_spanning_set}). 
 	\item Recover $D$ by either finding the smallest Spanning Set of $S$ (instance of Problem \ref{prob:subspace_arrangement_learning}),
 	or a fitted dictionary learning (instance of Problem~\ref{prob:fitted_dictionary_learning}).
 \end{itemize}
 
 Note that it is not true that the decomposition strategy should always be applied for the same sparsity $s$. 
 The decomposition starts out
 with the minimum given value of $s$ and is reapplied with iteratively higher $s$ if a solution
 has not be obtained.

 		 A summary of the input conditions and results of these different types of Dictionary Learning problems 
 		 from this section can be found in Table \ref{tab:classification}.

 		 \renewcommand{\arraystretch}{1.2}
 		 \begin{table*}[hptb]
 		 	\centering
 		 	\begin{footnotesize}
 		 		\resizebox{\textwidth}{!}{
 		 			\begin{tabular}{|p{0.12\textwidth}|p{0.1\textwidth}|p{0.1\textwidth}|p{0.125\textwidth}|p{0.1\textwidth}|p{0.23\textwidth}|p{0.165\textwidth}|p{0.18\textwidth}|}
 		 				\hline
 		 				\multirow{2}{*}{} & \multicolumn{4}{c|}{Traditional Dictionary Learning}  & 
 		 				\multirow{2}{\hsize}{Dictionary Learning via subspace arrangement and spanning set} & 
 		 				\multirow{2}{\hsize}{ Dictionary Learning for Segmented Data} & 
 		 				\multirow{2}{\hsize}{Fitted Dictionary Learning (this paper)}\\ \cline{2-5}
 		 				& Alternating Minimization Approaches & Spielman et. al \cite{spielman2013exact} & Arora et. al \cite{arora2013new}, Agarwal et. al \cite{agarwal2013exact} & For generic data (this paper) &  	&  &  \\ \hline
 		 				

 		 				\multirow{2}{\hsize}{Input and Conditions}  & \multirow{2}{\hsize}{$D$ satisfies frame property} & \multicolumn{2}{p{0.24\hsize}|}{$X$ generated from hidden dictionary $D$ and certain distribution of $\Theta$} &
 		 				 \multirow{2}{\hsize}{Generic data points $X$} &
 		 				 \multirow{2}{\hsize}{$X$ with promise that each subspace / dictionary support set is shared by sufficiently many data points in $X$} & \multirow{2}{\hsize}{Partitioned / segmented Data $X$} &
 		 				  \multirow{2}{\hsize}{$X$ with underlying hypergraph specified}   \\ \cline{3-4}
 		 				& &  $D$ is a basis &  $D$~is~pairwise incoherent & & & & \\ \hline

 		 				Minimum $m$ guaranteeing existence of a locally unique dictionary of a given size $n$ & Question~\ref{question:size_GDL} & $O(n \log n)$ & $O(n^2 \log^2 n)$ 
 		 				& $\dfrac{d-s}{d-1} n$ (Corollary~\ref{cor:dictionary_size_bound}); Unknown for general position data (Question~\ref{question:general_position})
 		 				& Minimum number of points to guarantee a unique subspace
 		 				arrangement that will give a spanning set of size $n$ & Question~\ref{question:support_equivalence_learning} & $\dfrac{d-s}{d-1} n$, where the underlying hypergraph satisfy Theorem~\ref{thm:rigidity_condition}
 		 				\\ \hline
 		 				
 		 				Dictionary Learning algorithms & MOD, k-SVD, etc. & Algorithm from \cite{spielman2013exact} & Algorithms from \cite{arora2013new,agarwal2013exact} & 
 		 				Straight\-forward algorithm (Corollary~\ref{corr:straight_forward_algo})&
 		 				Subspace Arrangement Learning Algorithms (Problem~\ref{prob:subspace_arrangement_learning}) and Spanning Set Finding ( Problem~\ref{prob:smallest_spanning_set}) & Question~\ref{question:support_equivalence_learning} and Spanning Set Finding ( Problem~\ref{prob:smallest_spanning_set}) & Realization using DR-plan (Theorem~\ref{theorem:canonical_is_optimal}) 
 		 				\\ \hline
 		 				
 		 				Minimum $m$ guaranteeing efficient dictionary learning & Unknown & \multicolumn{2}{p{0.2\hsize}	|}{$O(n^2 \log^2 n)$} &
 		 				Unknown & 
 		 				Unknown & Unknown  &Unknown \\ \hline
 		 				
 		 				Illustrative example & \multicolumn{4}{p{0.3\hsize}|}{}  & 
 		 				(a) \includegraphics[width=\hsize]{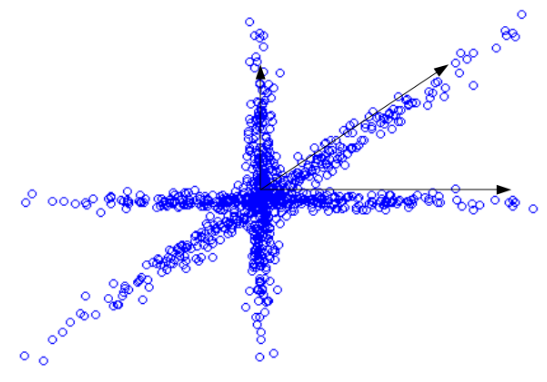}
 		 				& 
 		 				(b) \includegraphics[width=\hsize]{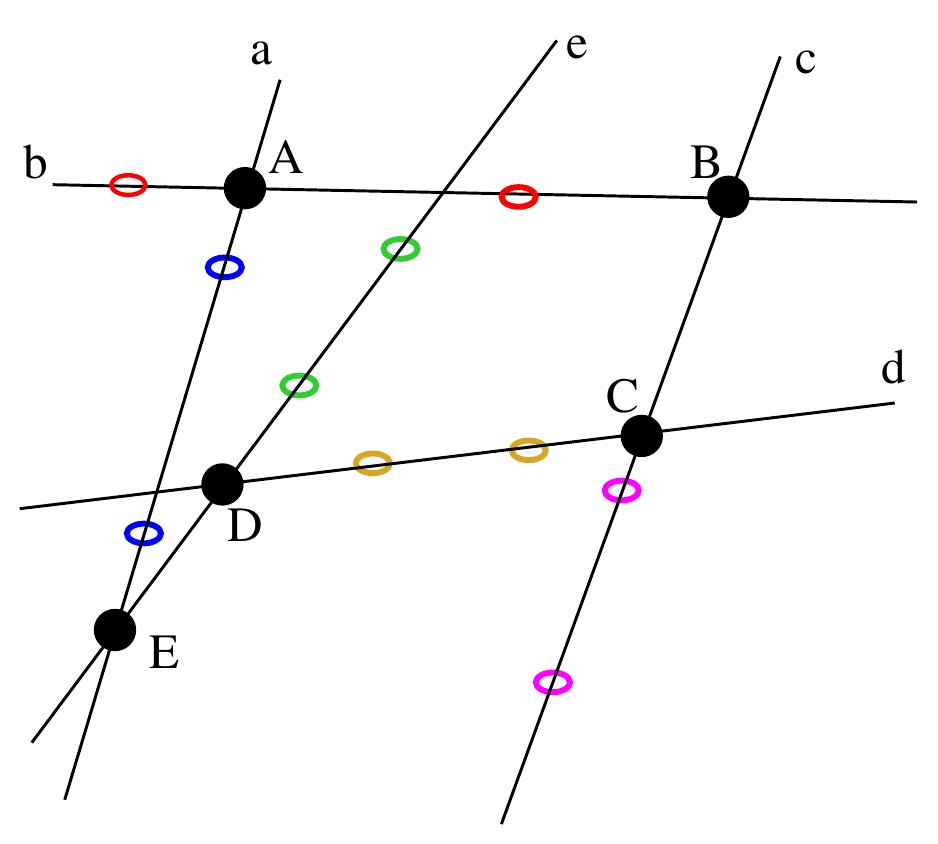}
 		 				& 
 		 				(c) \includegraphics[width=\hsize]{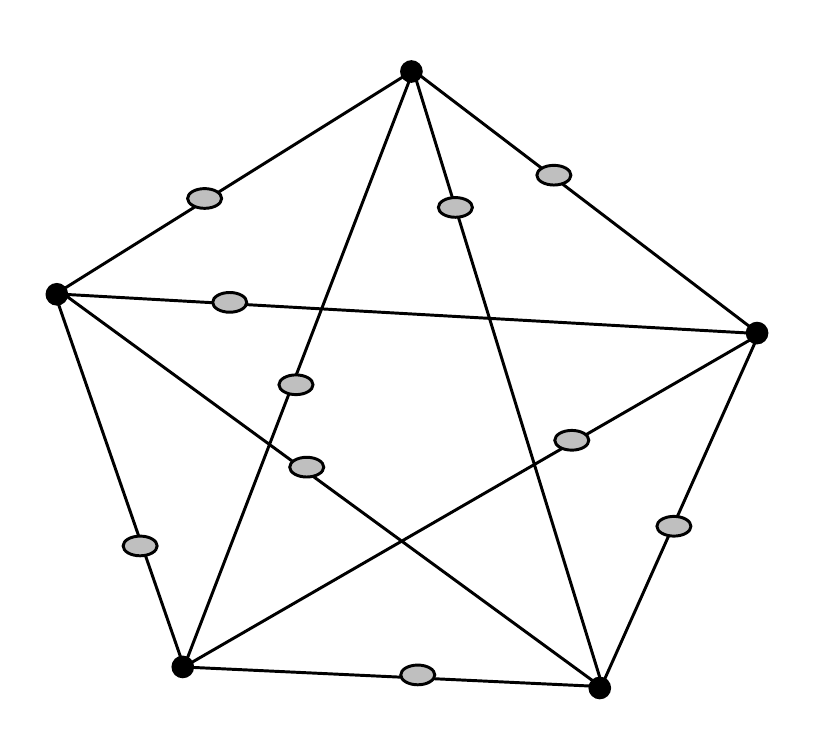}
 		 				
 		 				\\ \hline
 		 			\end{tabular}
 		 			}
 		 			
 		 		\end{footnotesize}
 		 		\caption{Classification of Problems}
 		 		\label{tab:classification}
 		 	\end{table*}

\subsection{New Bounds and Algorithms for Dictionary Learning for Random Data via Pinned Subspace-Incidence Systems}
\label{sec:bounds}

%

The following corollary of Theorem~\ref{thm:rigidity_condition}
gives a tight bound on dictionary size for generic data points,
which leads to an algorithm for finding a dictionary provided the bounds hold.
\begin{corollary}[Dictionary size tight bound for generic data]\label{cor:dictionary_size_bound_generic}
	Given a set of $m$ points $X=\{x_1,..,x_m\}$ in $\mathbb{R}^d$, generically there
	is a dictionary $D$ of size $n$ that $s$-represents $X$ 
	only if $(d-s)m \le (d-1)n$. 
	Conversely, if $(d-s)m = (d-1)n$ and the supports of $x_i$ (the nonzero entries of
	the $\theta_i$'s) are known to form a $(d-1,0)$-tight hypergraph $H$, then generically, there
	is at least one and at most finitely many such dictionaries.
\end{corollary}

%

Quantifying the term ``generically'' in Corollary~\ref{cor:dictionary_size_bound_generic}
yields Corollaries~\ref{cor:dictionary_size_bound} and \ref{corr:straight_forward_algo} below.  

\begin{corollary}[Lower bound for random data]\label{cor:dictionary_size_bound}
	Given a set of $m$ points $X=\{x_1,..,x_m\}$ picked from a distribution $\rho$
	with respect to which nongenericity has measure zero, 
	a dictionary $D$ that $s$-represents $X$ has size at least $(\frac{d-s}{d-1})m$ with probability 1.
	In other words, 
	$|D| = \Omega(X)$ if $s$ and $d$ are constants.
\end{corollary}



\begin{corollary}[Straightforward Dictionary Learning Algorithm]\label{corr:straight_forward_algo}
	Given a set of $m$ points  $X=[x_1 \ldots x_m]$ picked from a distribution $\rho$ 
	with respect to which nongenericity has measure zero, 
	there is a straightforward pebble-game~\cite{jacobs1997algorithm,lee2005finding,lee2007pebble} based algorithm constructs a dictionary $D=[ v_1 \ldots v_n ]$ that $s$-represents $X$,
	where $n = \left( \dfrac{d-s}{d-1} \right) m$.
	The time complexity of the algorithm is $O(m)$ when we treat $d$ and $s$ as constants. 
\end{corollary}

The key idea of the algorithm for Corollary~\ref{corr:straight_forward_algo} is that we can choose a convenient underlying hypergraph $H(S_{X,D})$.
Thus the algorithm has two parts: (1) constructing $\hsxd$ 
and (2) constructing the $s$-subspace arrangement $\sxd$ and the dictionary $D$. 

\smallskip\noindent \textbf{(1) Algorithm for
	constructing the underlying hypergraph $H(S_{X,D})$ for a hypothetical
	$s$-subspace arrangement $S_{X,D}$}:

The algorithm  works in three stages to construct a expanded mutli-hypergraph $\hat{H}(S_{X,D})$:
\begin{enumerate}
	\item 
	We start by constructing a minimal minimally rigid hypergraph $H_0 = (V_0,E_0)$, using the \emph{pebble game algorithm} introduced below.
	Here $|V_0| = k(d-s)$, $|E_0| = k(d-1)$, where $k$ is the smallest positive integer such that
	$ {k(d-s) \choose s}  \ge k(d-1) $,
	so it is possible to construct $E_0$ such that
	no more than one hyperedge in $E_0$ containing the same set of vertices in $V_0$.
	The values $|V_0|$ and $|E_0|$ are constants when we think of $d$ and $s$ as constants.
	
	\item We use the pebble game algorithm to append a set $V_1$ of $d-s$ vertices and a set $E_1$ of $d-1$ hyperedges 
	to $H_0$, such that each hyperedge in $E_1$ contains at least one vertex from $V_1$, 
	and the obtained graph $H_1$ is still minimally rigid. 
	The subgraph $B_1$ induced by $E_1$ has vertex set $V_{B_1} = V_1  \bigcup V_B$, 
	where $V_B \subset V_0$. 
	We call the vertex set $V_B$ the \emph{base vertices} of the construction.
	
	\item Each of the following construction step $i$ appends
	a set $V_i$ of $d-s$ vertices and a set $E_i$ of $d-1$ hyperedges
	such that the subgraph $B_i$ induced by $E_i$ has vertex set $V_i \bigcup V_B$,
	and $B_i$ is isomorphic to $B_1$. 
	In other words, at each step, 
	we directly append a basic structure the same as $(V_1,E_1)$ to the base vertices $V_B$. 
	It is not hard to verify that the obtained graph is still minimally rigid.

	%
\end{enumerate}
The \emph{pebble game algorithm} by \cite{streinu2009sparse} works on a fixed finite set $V$ of vertices 
and constructs a $(k,l)$-sparse hypergraph. 
Conversely, any $(k,l)$-sparse hypergraph on vertex set $V$ can be constructed by this algorithm.
The algorithm initializes by putting $k$ pebbles on each vertex in $V$. 
There are two types of moves: 
\begin{itemize}
	\item \emph{Add-edge}: adds a hyperedge $e$ (vertices in $e$ must contain at least $l+1$ pebbles), removes a pebble from a vertex $v$ in $e$, and assign $v$ as the tail of $e$; 
	\item \emph{Pebble-shift}: for a hyperedge $e$ with tail $v_2$, and a vertex $v_1 \in e$ which containing at least one pebble, 
	moves one pebble from $v_1$ to $v_2$, and change the tail of $e$ to $v_1$.
\end{itemize}
At the end of the algorithm, if there are exactly $l$ pebbles in the hypergraph, then the hypergraph is $(k,l)$-tight.

Our algorithm runs a slightly modified pebble game algorithm to find a $(d-1,0)$-tight expanded mutli-hypergraph.
We require that each add-edge move adding  $(d-s)$ copies of a hyperedge $e$, 
so a total of $d-s$ pebbles are removed from vertices in $e$.
Additionally, the multiplicity of a hyperedge, not counting the expanded copies, cannot exceed 1. 
For constructing the basic structure of Stage 2, 
the algorithm initializes by putting $d-1$ pebbles on each vertex in $V_1$.
In addition, an add-edge move can only add a hyperedge that contains at least one vertex in $V_1$, 
and a pebble-shift move can only shift a pebble inside $V_1$.

\sloppy
The pebble-game algorithm takes $O\left(s^2 |V_0| {|V_0| \choose s}\right)$ time in  Step 1, 
and $O\left(s^2 \left(|V_0| + (d-s)\right) {|V_0| + (d-s) \choose s }\right)$ time in Step 2. 
Since the entire underlying hypergraph $\hsxd$ has $m = |X|$ edges, 
Step 3 will be iterated $O( m / (d-1))$ times, and each iteration takes constant time. 
Therefore the overall time complexity for constructing $\hsxd$ is 
\[O\left( s^2 \left(|V_0| + (d-s)\right) {|V_0| + (d-s) \choose s } + \left(m/(d-1)\right) \right)\] 
which is $O(m)$ when $d$ and $s$ are regarded as constants. 

\fussy

\smallskip\noindent \textbf{(2) Algorithm for constructing the
	$s$-subspace arrangement $S_{X,D}$ and the dictionary $D$}:

The construction of the $s$-subspace arrangement $S_{X,D}$ naturally follows from 
the construction of the underlying hypergraph $\hsxd$.
For the initial hypergraph $H_0$, 
we get a pinned subspace-incidence system $(H_0,X_0)(D_0)$ by arbitrarily choose $|X_0| = |E_0|$ pins from $X$.
Similarly,  for Step 2 and each iteration of Step 3,
we form a pinned subspace-incidence system $(B_i,X_i)(D_i)$ by arbitrarily choosing $|X_i| = d-1$ pins from $X$.

Given $X_0$, we know that the rigidity matrix
-- of the $s$-subspace framework $H_0(S_{X_0,D_0})$ -- with indeterminates representing
the coordinate positions of the points in $D_0$ -- generically has full rank
(rows are maximally independent), under the pure conditions of Theorem \ref{thm:rigidity_condition};
in which case, the original algebraic subsystem $(H_0,X_0)(D_0)$ 
(whose Jacobian is the rigidity matrix), with $X_0$ plugged in,
is guaranteed to have a (possibly complex) solution  and only finitely
many solutions for $D_0$. Since the pure conditions fail only on a
measure-zero subset of the space of pin-sets $X_0$, where each pin is in
$S^{d-1}$, it follows that if the pins in $X_0$ are picked uniformly at random from $S^{d-1}$
we know such a solution exists for $D_0$ (and $S_{X_0,D_0}$) and can be found by
solving the algebraic system $H_0(S_{X_0,D_0})$.

Once we have solved $(H_0,X_0)(D_0)$, 
for each following construction step $i$, 
$B_i$ is also rigid since coordintate positions of the vertices in $V_B$ have been fixed
So similarly, we know a solution exists for $D_i$ (and $S_{X_i,D_i}$) and can be found by
solving the algebraic system $B_i(S_{X_i,D_i})$, which is of constant size $O(d)$.
Although there can be more than one choice of solution for each step, 
since every construction step is based on base vertices $V_B$, 
the solution of one  step will not affect any other  steps, 
so generically any choice will result in a successful solution for the entire construction sequence, 
and we obtain $D$ by taking the union of all $D_i$'s.

When we regard $d$ and $s$ as constants, 
the time complexity for Stage (2)
is the constant time for solving the size $O(|V_0|)$ algebraic system $(H_0,X_0)(D_0)$, 
plus $O(m / (d-1))$ times the constant time for solving the size $O(d)$ system $(B_i,X_i)(D_i)$, 
that is $O(m)$ in total. 

Therefore the overall time complexity of the dictionary learning algorithm is $O(m)$. 

\subsection{Dictionary Learning for Commonly Occurring Data Using DR-Planning}

In practical dictionary learning problems, the data set is usually overconstrained, making it impossible to apply the algorithm given above.

As one solution for such dictionary learning problems,
recall the two-step procedure described in Section~\ref{sec:review}, 
where we first use standard preprocessing algorithms such as RANSAC and GPCA to learn a subspace arrangement from the data set (Problem~\ref{prob:subspace_arrangement_learning}),
and then constructing a fitted dictionary learning problem (Problem~\ref{prob:fitted_dictionary_learning}) based on the subspace arrangement learnt.
Now we can obtain a dictionary by realizing the (possibly overconstrained) pinned subspace incidence-system corresponding to the fitted dictionary learning problem. 

To find a realization to a geometric constraint system, 
the straightforward method is to directly find the real solutions to the entire multivariate polynomial system. 
However, such an approach requires double exponential time in the number of variables.
Thus, it is crucial to use recursive \emph{Decomposition-Recombination (DR-) plans} \cite{hoffman2001decompositionI,hoffman2001decompositionII,jermann2006decomposition,sitharam2005combinatorial} ,
which decompose the original system into locally rigid subsystems and recombines solutions of subsystems to get a solution for the original system.
For a pinned subspace-incidence system $(H, Q)$, a DR-plan is formally defined for the underlying multi-hypergraph $\hat{H}$ (recall the definition from Section~\ref{sec:graph}) as following:

\begin{definition}[DR-plan]
	A \emph{decomposition-recombination (DR-)plan} of the underlying multi-hypergraph $\hat{H}$ of a pinned subspace-incidence system $(H, Q)$ in $\R$ is a forest that has the following properties:
	\begin{enumerate}
		\item\ Each node represents a connected vertex-induced subgraph of $\hat{H}$ that is either  rigid, or a single hyperedge.
		\item\ A root node represents a maximal subgraph of $\hat{H}$.
		\item\ A node represents the subgraph of $\hat{H}$ induced by the union of the node's children. 
		\item\ A leaf node represents a single hyperedge.
	\end{enumerate}
\end{definition}

With the use of a DR-plan, the complexity of realizing a geometric constraint system is determined by the maximum \emph{fan-in} of the DR-plan,
i.e.\ the maximum number of children nodes of any node.
Finding an \emph{optimal DR-plan} i.e.\ one that minimizes the maximum fan-in, however, is usually hard for general geometric constraint systems.

In the recent work~\cite{baker2015designing}, it is shown that an optimal DR-plan can be efficiently found for a large class
of minimally rigid geometric constraint systems with $O(n^3)$ time complexity, leading to efficient realization.
Specifically, this result holds for minimally rigid pinned subspace-incidence systems, as stated in the following theorem:

\begin{theorem}[\cite{baker2015designing}]
	\label{theorem:canonical_is_optimal}
	For generically minimally rigid pinned subspace-incidence, there exists a DR-plan of the underlying multi-hypergraph satisfying the following two additional properties, and any such DR-plan is optimal:
		\begin{enumerate}
			\item\ Children are connected, generically rigid, vertex-maximal proper subgraphs of the parent.
			\item\ If all pairs of rigid vertex-maximal proper subgraphs intersect trivially then all of them are children, otherwise exactly two that intersect non-trivially are children.
		\end{enumerate}
\end{theorem}

Figure~\ref{fig:well} shows an example of the optimal DR-plan found by applying Theorem~\ref{theorem:canonical_is_optimal} to a generically minimally rigid pinned subspace-incidence system
with 12 vertices and 24 pins, where $d=3$.
The DR-plan has the optimal maximum fan-in  12.

\begin{figure}
	\centering
		\begin{subfigure}{.4\linewidth}
			\centering
			\includegraphics[width=.8\linewidth]{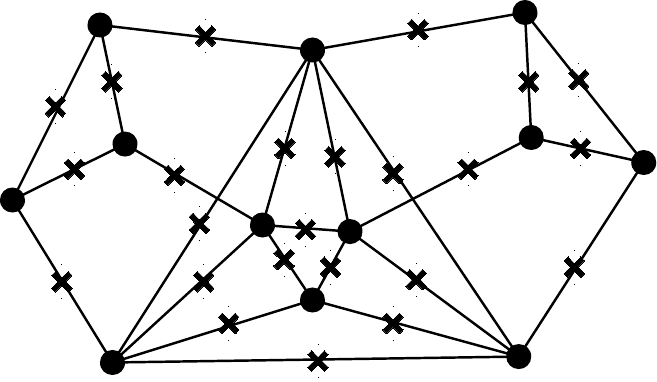}
			\caption{}
			\label{fig:well_system}
		\end{subfigure}%
		\begin{subfigure}{.6\linewidth}
			\centering
			\includegraphics[width=.8\linewidth]{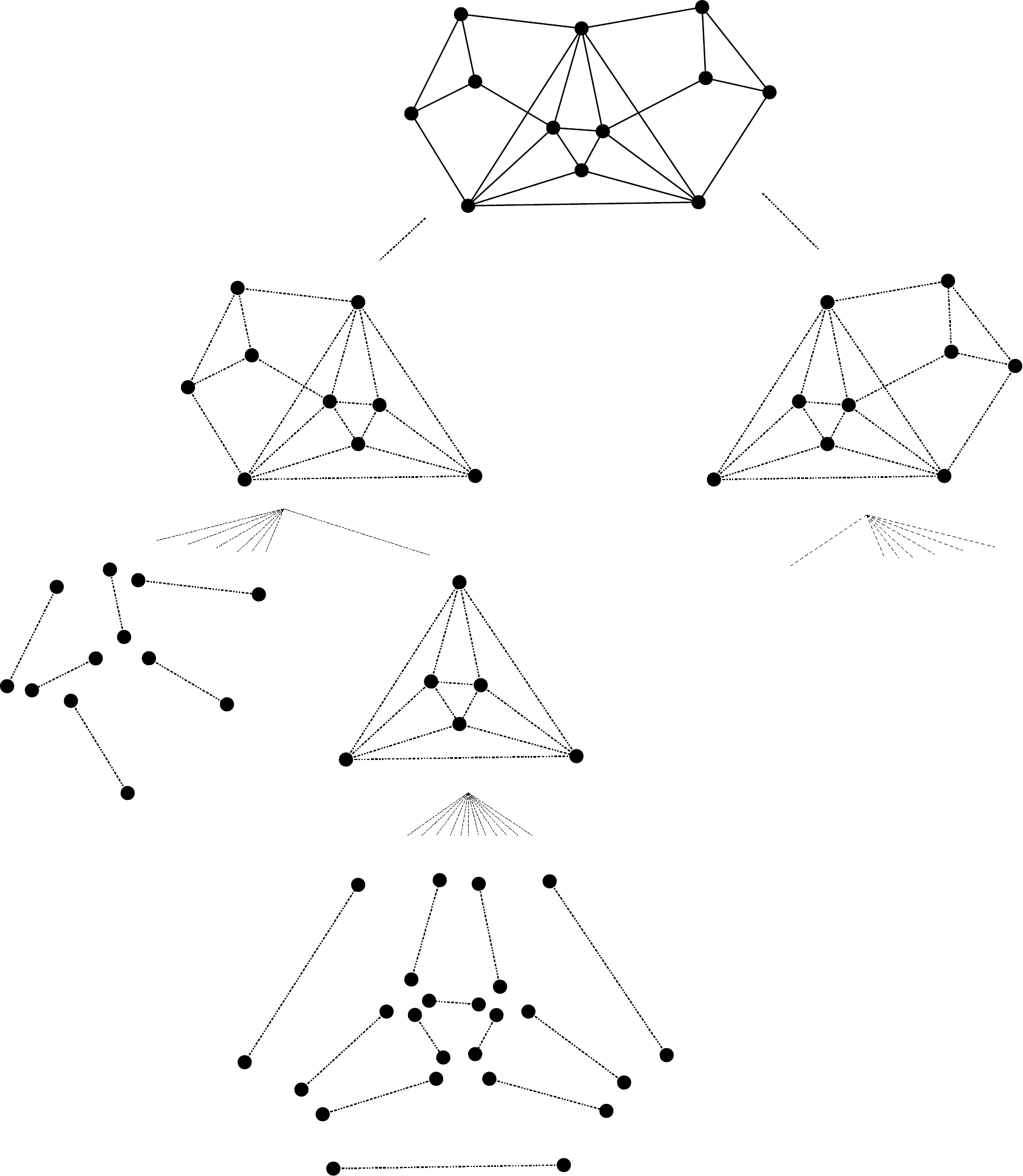}
			\caption{}
			\label{fig:well_plan}
		\end{subfigure}
		\caption{ (a) A generically minimally rigid pinned subspace-incidence system with $d=3$. (b) The optimal DR-plan of (a) obtained using Theorem~\ref{theorem:canonical_is_optimal}.}
	\label{fig:well}
\end{figure}

However, for overconstrained systems, the optimal DR-planning problem
was shown to be NP-hard even for bar-joint systems \cite{lomonosov2004graph},
and we strongly conjecture it is also NP-hard for pinned subspace-incidence systems.
One possible workaround is to first apply the pebble-game algorithm from \cite{jacobs1997algorithm,lee2005finding} 
to an overconstrained system $(H,Q)$ to find a maximal, generically minimally rigid subsystem, and then apply the algorithm from~\cite{baker2015designing} to find an optimal DR-plan for the subsystem.
However, by doing this we are forgoing the advantages of smaller sized DR-plan which is usually provided by overconstrained cases.
We also note that different maximal, generically minimally rigid subsystems of $(H,Q)$ may still have different maximum fan-ins.
As an example, Figure~\ref{fig:over}(a) shows a overconstrained pinned subspace-incidence system with $d=3$. 
Figure~\ref{fig:over} (b) and (c) are DR-plans of two maximal, generically minimally rigid subsystems of (a), where (b) has the optimal maximum fan-in 11, while (c) has a maximum fan-in of 12.

\begin{figure}
	\centering
	\begin{subfigure}{.23\linewidth}
		\centering
		\includegraphics[width=\linewidth]{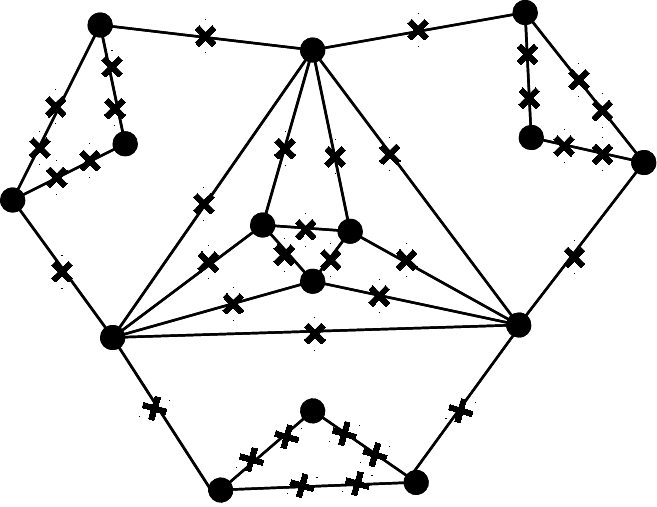}
		\caption{}
		\label{fig:over_system}
	\end{subfigure}%
	\begin{subfigure}{.4\linewidth}
		\centering
		\includegraphics[width=.95\linewidth]{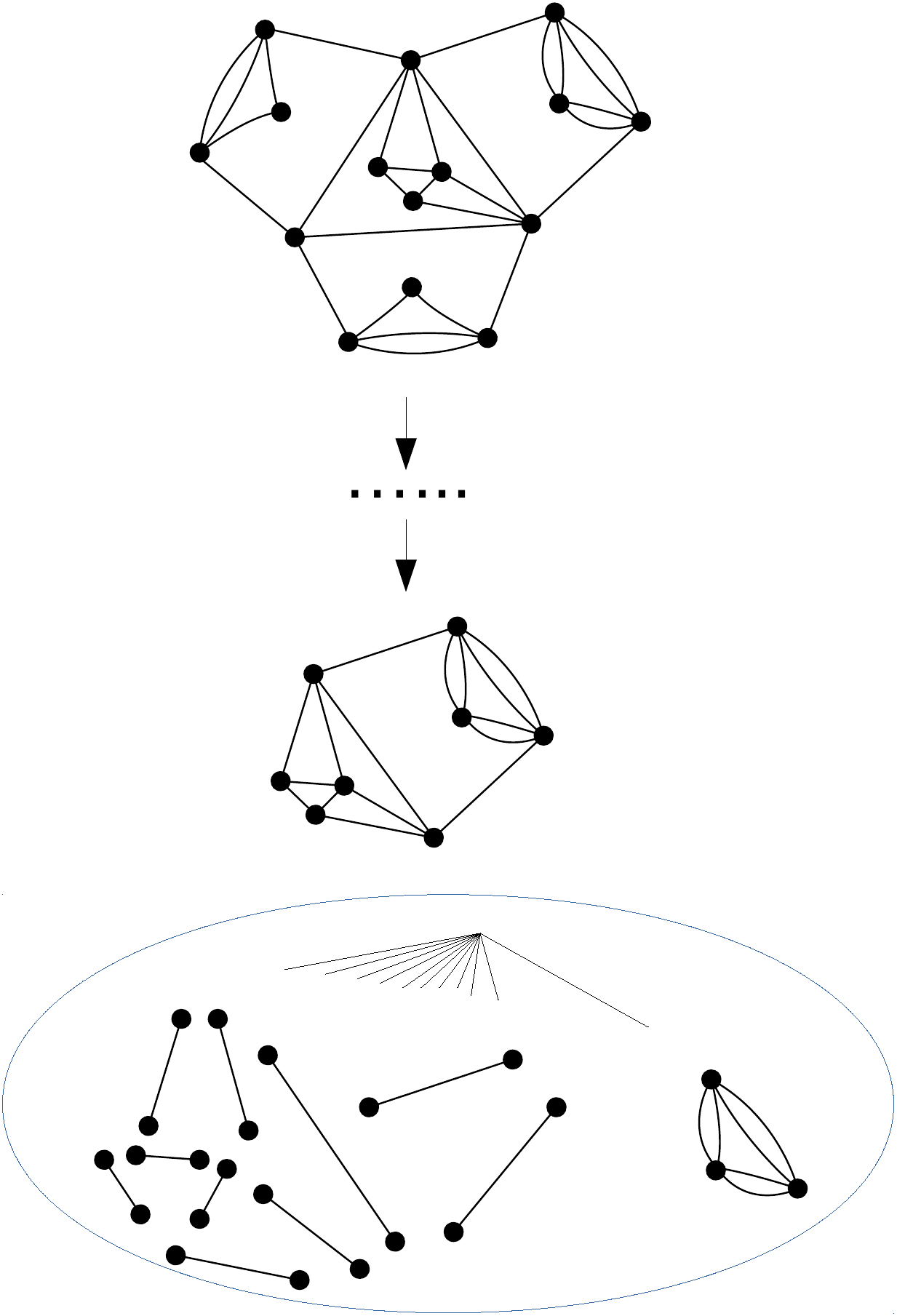}
		\caption{}
		\label{fig:over1}
	\end{subfigure}%
	\begin{subfigure}{.37\linewidth}
		\centering
		\includegraphics[width=.95\linewidth]{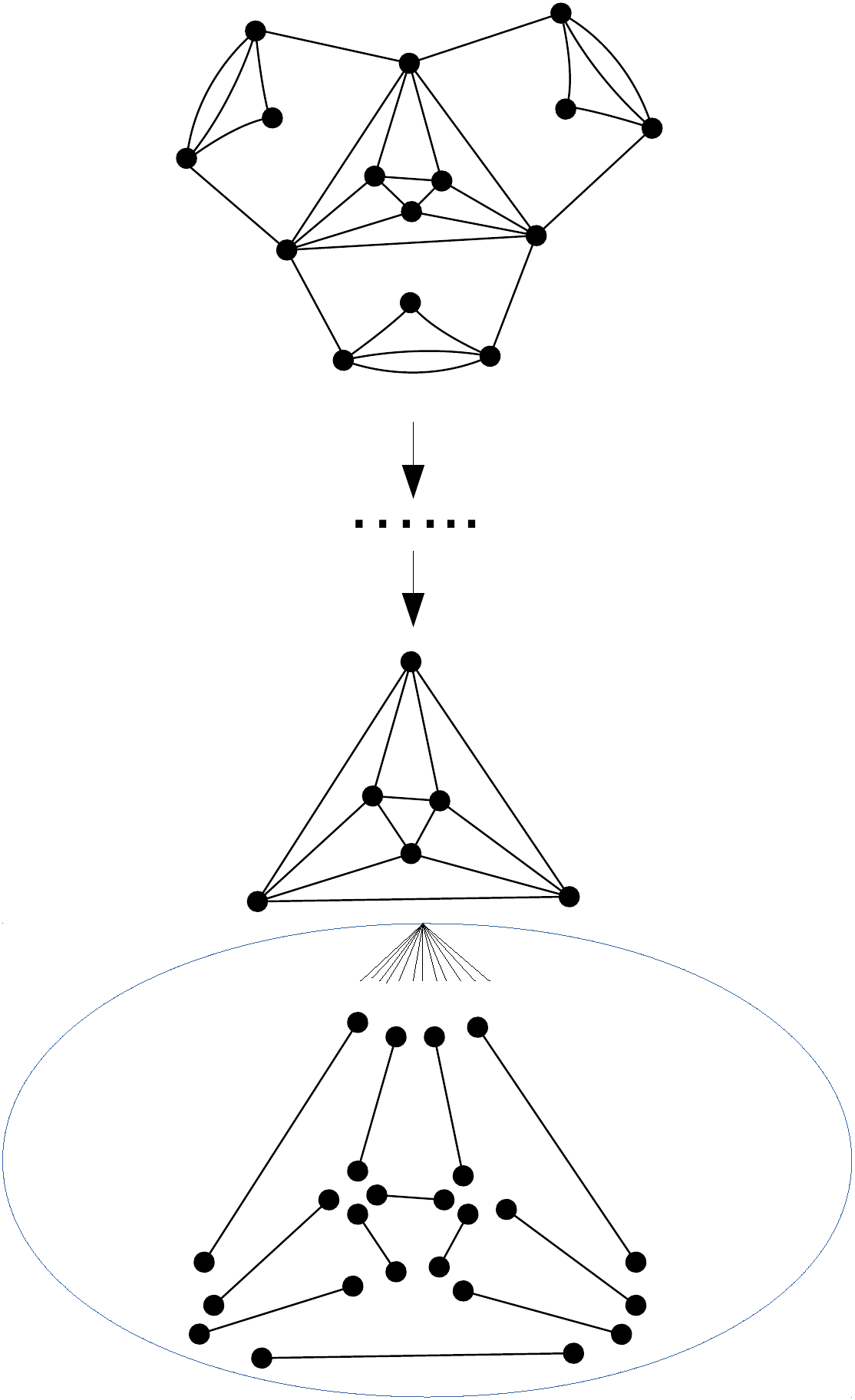}
		\caption{}
		\label{fig:over2}
	\end{subfigure}
	\caption{ (a) A overconstrained pinned subspace-incidence system with $d=3$. (b) A maximal minimally rigid subgraph with maximum fan-in 11 (optimal). (c)  A maximal minimally rigid subgraph with maximum fan-in 11}
	\label{fig:over}
\end{figure}

\section{Open Questions}
A natural restriction of the two step problem in Section~\ref{sec:review}, which learns subspace arrangement followed by spanning set is the following,
where the data set $X$ is given in \emph{support-equivalence} classes. 
For a given subspace $t$ in the subspace arrangement $\sxd$ (respectively hyperedge $h$ in the hypergraph's edge-set $\calsxd$), 
let $X_t = X_h  \subseteq X$ be the equivalence class of data points $x$ such that $\suppx{x} = t.$ 
We call the data points $x$ in a same $X_h$ as \emph{support-equivalent}.
\begin{question} [Dictionary Learning for Partitioned Data] \label{question:support_equivalence_learning}
	(1) What is the minimum size of $X$ and $X_i$'s (representing data $X$ partitioned into $X_i's$) guaranteeing 
	that there exists a locally unique dictionary $D$ 
	for a $s$-subspace arrangement $\sxd$ satisfying $|D| \le n$, 
	and $X_i$ represents the support-equivalence classes of $X$ with respect to $D$?
	
	(2) How to find such  a dictionary $D$? 
	\end{question}

	With regard to the problem of minimizing $|D|$,  
	very little is known for simple restrictions on $X$. 
	For example the following question is open.
	\begin{question} \label{question:general_position}
		Given a general position assumption on $X$,
		what is the best lower bound on $|D|$ for Dictionary Learning?
		Conversely, are smaller dictionaries possible than indicated by 
		Corollary~\ref{cor:dictionary_size_bound} (see Section~\ref{sec:bounds}) under such an assumption?
		\end{question}

		Question \ref{question:general_position} gives rise to the following pure combinatorics open question
		closely related to the intersection semilattice of subspace arrrangement \cite{bjorner1994subspace,goresky1988stratified}.
		
		\begin{question}
			\label{question:lattice}
			Given weights $w(S)\in \mathbb{N}$ assigned to size-$s$ subsets $S$ of $[n]$. 
			For $T \subseteq [n]$ with $|T| \ne s$, 
			\[
			w(T) = 
			\begin{cases}
			0 & |T| < s\\
			\displaystyle\sum_{S \subset T, |S| = s} w(S) & |T|> s 
			\end{cases}
			\]
			Assume  additionally the following constraint holds:
			for all subsets $T$ of $[n]$ with $s \le |T| \le d$, $w(T) \le |T|-1$.
			Can one give a nontrivial upper  bound on $w([n])$?
			\end{question}

			The combinatorial characterization given by Theorem \ref{thm:rigidity_condition} leads to the following question
			for general Dictionary Learning.
			\begin{question}
				\label{question:size_GDL}
				What is the minimum size of a data set $X$ such that 
				the Dictionary Learning for $X$ has a locally unique
				solution dictionary $D$ of a given size? 
				What are the geometric characteristics of such an $X$? 
				\end{question}

\section{Conclusion}

In this paper, we studied the generic rigidity of pinned subspace-incidence systems, and  obtained a combinatorial characterization of  generic minimal rigidity for pinned subspace-incidence systems.
We then showed the application of pinned subspace-incidence systems  in Dictionary Learning, 
giving a tight bound on dictionary size in a specific setting as well as new dictionary learning algorithms in this and other general settings.
We provided a systematic classification of problems related to dictionary learning 
together with various algorithms, their assumptions  and performance, 
and formalized related open questions from a geometry perspective.

 	\bibliographystyle{siam}
 	\bibliography{refs}
\end{document}